\documentclass[journal]{IEEEtran}

\usepackage{graphicx}
\usepackage{amsmath}
\usepackage{cases}
\usepackage{array}
\usepackage{subcaption}
\usepackage{array}
\usepackage{lipsum,color}
\usepackage{epstopdf}
\usepackage{graphicx}
\usepackage{wrapfig}
\usepackage{mathrsfs}
\usepackage{tabularx}
\usepackage{multirow} 
\usepackage{graphicx}
\usepackage{cite}
\usepackage{hyperref}
\usepackage{algorithm,algorithmic,amsmath}
\usepackage{mathtools}
\usepackage{dsfont}
\usepackage{amsthm,amssymb} 
\usepackage{textcomp}

\DeclareMathOperator{\E}{\mathbb{E}}

\newtheorem{theorem}{Theorem}
\newtheorem{proposition}{Proposition}

\DeclareGraphicsExtensions{.pdf,.png,.jpg,.eps,.ps}

\begin{document}	
	\title{Channel Estimation, Interference Cancellation, and Symbol Detection for 	Communications on Overlapping Channels}	
	\author{Minh Tri Nguyen and Long Bao Le\\
		INRS, University of Quebec, Montreal, QC, Canada. E-mails: \{minh.tri.nguyen, long.le\}@emt.inrs.ca
		\thanks{The article has been accepted for publication in IEEE Access in May 2020.}
        \thanks{{\sffamily\textcopyright} 2020 IEEE.  Personal use of this material is permitted.  Permission from IEEE must be obtained for all other uses, in any current or future media, including reprinting/republishing this material for advertising or promotional purposes, creating new collective works, for resale or redistribution to servers or lists, or reuse of any copyrighted component of this work in other works.}
        \thanks{The DOI of the accepted article is: 10.1109/ACCESS.2020.2993582}}
	\maketitle 
	
\begin{abstract}
In this paper, we propose the joint interference cancellation, fast fading channel estimation,
and data symbol detection for a
general interference setting where the interfering source and
the interfered receiver are unsynchronized and
occupy overlapping channels of different bandwidths.
The interference must be canceled before the channel
estimation and data symbol detection of the desired communication are performed.
To this end, we have to estimate the Effective Interference Coefficients (EICs)
and then the desired fast fading channel coefficients.
We construct a two-phase framework
where the EICs and desired channel coefficients are estimated
using the joint maximum likelihood-maximum a posteriori probability
(JML-MAP) criteria in the first phase;
and the MAP based data symbol detection is performed in the second phase.
Based on this two-phase framework, we also propose an iterative algorithm
for interference cancellation, channel estimation and data detection.
We analyze the channel estimation error, residual interference,
symbol error rate (SER) achieved by the proposed framework.
We then discuss how to optimize the pilot density to achieve the maximum throughput.
Via numerical studies, we show that our design can effectively mitigate the interference
for a wide range of SNR values,
our proposed channel estimation and symbol detection design
can achieve better performances compared to the existing method.
Moreover, we demonstrate the improved performance of the iterative algorithm 
with respect to the non-iterative counterpart.
\end{abstract}

 
\IEEEpeerreviewmaketitle
 
\section{Introduction}
\label{sec:introduction}

Traffic
 demand from wireless networks has been increasing dramatically
over the last decades while
the
spectrum resource is limited.
This has motivated
the
development of efficient and flexible spectrum utilization
and sharing techniques.
Moreover, future wireless networks are expected to support
a massive number of connections to enable many emerging applications
requiring diverse communication rates and
qualities
of service \cite{naik2018coexistence}.
Therefore, effective spectrum reuses using robust interference cancellation
and management are essential in maintaining and enhancing
the communication rates and reliability in next-generation wireless systems \cite{lee2016advanced}.
In particular, future wireless systems
must be able to support
different applications and use cases,
e.g., highly mobile scenarios
in which users move at high speeds (up to 500 km/hr)
\cite{5GvisionSamsung, 5GvisionEricsson, shah20185g}.
Thus, developing wireless communication techniques for high mobility environments is of high importance and
has
attracted increasing research attention
\cite{3GPP_15, wu2016survey, noh2019realizing}.

Non-Orthogonal Multiple Access (NOMA) \cite{cai2017modulation} and 
Full Duplex (FD) communication \cite{bharadia2013full} are among the advanced frequency reuse techniques.
In NOMA, signals from different sources are allowed to be transmitted simultaneously over
the same channel, and successive interference cancellation is typically employed to decode these messages.
Moreover, a FD transceiver allows to transmit and receive at the same time over the same channel, thus, 
the receiver experiences severe self-interference from the transmitter.
As a result, advanced interference  cancellation techniques are required to realize a practical FD system
where combined analog and digital interference cancellation strategies are usually employed to achieve sufficient 
cancellation performance \cite{sabharwal2014band}.

Note, however, that FD communication has a special interference structure
where the interfering and interfered communications have the same bandwidth (hence, the same symbol rate). 
This interference structure plays a crucial role in designing interference cancellation techniques, especially in 
the digital domain \cite{ahmed2015all,korpi2014widely}.
Interference cancellation in the more general scenario where the interfering source and victim
have different bandwidths is more challenging to tackle
because of the following reasons.
First, the equivalent interference
coefficients (EIC) \cite{Nguyen2017} vary from symbol to symbol and they are
difficult to capture. Second, when operating over different
bandwidths, these concurrent communications are likely not synchronized, which
creates a fundamental limitation in cancellation performance \cite{syrjala2014analysis}.

Various interference cancellation techniques, including
passive interference cancellations \cite{everett2014passive}, active interference cancellations in the analog domain 
\cite{duarte2012experiment, debaillie2014analog} and in the digital domain \cite{korpi2014widely, ahmed2015all},
have been proposed for full-duplex systems. However, only a few works study interference cancellation
for the concurrent communications with different bandwidths even though
this interfering scenario can arise in both terrestrial communications \cite{pedersen2018agile}
	and satellite communications \cite{FCC_sat}.
In fact, this interference scenario occurs between the Iridium satellite system operating in the band 1621.35 - 1626.5 MHz and the
	Inmarsat satellite system operating in the adjacent band 1626.5 - 1660.5 MHz, as reported in \cite{euro_report}.
Thus, development of robust interference cancellation methods that 
effectively address the general interference scenario between two communications of
	different bandwidths is highly important.

Interference cancellation for communications with different bandwidths  
has been investigated in some previous works \cite{schwarzenbarth2007,Nguyen2017} assuming perfect CSI
and/or synchronization between the underlying communications.
The problem becomes much more challenging
when the desired channel experiences the fast fading
where the time-varying channel can be modeled by using the Gauss-Markov process \cite{wang1996verifying, bottomley1998unification, tan2000first, zhang2013blind, wu2016survey, zhu2009message}.
For the fast fading channel,
MMSE-based channel estimators are derived in \cite{sun2014maximizing, ma2003optimal}
requiring the knowledge of the channel correlation matrix,
which may not be readily usable
in the presence of interference.
Therefore, it is highly desirable to develop robust interference cancellation techniques that can effectively
cope with a strong interfering signal with different bandwidth from the victim in the fast fading environment.

Data symbol detection in the fast fading environment is another challenging task, especially with the presence of strong
 interference.
 A well-known approach for symbol detection in fast fading environments is
 the message-passing detection technique in which 
the posterior probability of data symbols is estimated.
In \cite{zhu2009message}, it is shown that this detection technique can function well
if the interfering signal has similar characteristics with the desired signal.
However, the method works well only
if the interfering and desired signals are synchronized and have the same symbol rate.
Furthermore, an approximated distribution of data symbols
by the Gaussian mixture with a limited number of terms
may yield unacceptable error rate with a large signal constellation size.
Another approach is considered in  \cite{mahamadu2018fundamental} where the channel gains at data symbols are interpolated 
by the imperfect CSI at pilot symbols.
Then, the zero-forcing based symbol detection is employed,
the technique is called optimum diversity detection (ODD).
However, this detection technique does not fully exploit the correlations of channel gains at consecutive data symbols,
and the required inverse matrix operations result in high computational complexity.
This motivates us to develop a new detection strategy that
has low complexity and can achieve the performance close to that of the ODD technique.

The above survey suggests that joint
channel estimation, interference cancellation, and symbol detection
for the scenario in which
two un-synchronized mutual interfering signals have different bandwidths in the fast fading environment
has been under-explored. This paper aims to fill this gap in the literature where we make the following contributions.

\begin{itemize}
	\item  Firstly, a two-phase framework for joint interference cancellation, channel estimation, and symbol detection is proposed.
	In the first phase, the EICs are estimated and the interference is subtracted.
	Then, fast-fading channel coefficients at pilot positions are estimated.
	In the second phase, 
	we derive the \textit{a posteriori probabilities} for both series and individual symbols, given the channel coefficients at pilot positions.
	Based on these probabilities, we propose corresponding detection methods.
	Specifically, our series symbol detection (S-MAP) outperforms the existing ODD technique \cite{mahamadu2018fundamental}
	while our individual symbol detection (I-MAP) achieves almost identical result to the ODD technique with much lower
	complexity as confirmed by numerical studies.
	
	\item Secondly, based on the proposed two-phase framework, we propose an iterative algorithm for 
	interference cancellation, channel estimation, and data detection. Numerical studies show
	that the proposed iterative algorithm converges quite quickly and it performs better than the non-iterative counterpart.
		
	\item Thirdly, we analyze the residual interference and symbol error rate achieved by the proposed non-iterative algorithm.
	Specifically, we provide an exact expression for channel estimation error in the interference-free scenario,
	and an approximated residual interference and channel estimation error for the case with interference.
	The analysis shows that the residual interference has bounded power as the interference power tends to infinity.
	However, the effect of the fast fading channel to the residual interference is irreducible
	no matter how large the SNR or the number of pilot symbols is.
	Hence, there are fundamental  floors for the channel estimation and symbol detection performances.
	
	\item Finally, we conduct simulation studies and draw
	several insightful observations from the results. 
	Particularly, the performance floor exists for the considered interference scenario while 
	it is not the case for the interference free scenario.
	It is also shown that the existing symbol detection method
	may need more than 3dB increment in SNR to achieve the same symbol error rate (SER) obtained from our S-MAP method,
	while our I-MAP method achieves very close performance to the existing optimum detection method.
	Finally, we show that there exists an optimal frame structure (i.e., optimal pilot density) to achieve the 
	maximum system 	throughput.
\end{itemize}

While preliminary results of this paper were published in \cite{nguyen2018channel},
the current paper makes several significant contributions
compared to this conference version.
Specifically, the current journal paper proposes two detection methods
with improved performances compared to
the method introduced in the conference version.
The new iterative algorithm is also proposed in this journal version.
The theoretical performance analysis and throughput optimization
were not conducted in the conference version.
Moreover, the current journal paper presents much more extensive numerical results
which provide useful insights into the proposed design.

The paper is structured as follows.
The system model and problem formulation are presented in Section II.
Section III describes the proposed channel estimation, interference cancellation, and 
the symbol detection techniques.
In Section IV, we analyze the residual interference, SER, and optimal frame design for the fast fading and interference scenario.
Numerical results are presented in Section V and Section VI concludes the paper.

Some important notations used in the paper are summarized as follows:
$\mathbf{I}_N$ represents the $N \times N$ identity matrix,
$\mathbf{1}_{M,N}$ is the $M \times N$ all-one matrix,
$\mathbf{A}^H$ is the \textit{Hermitian transpose} of matrix $\mathbf{A}$,
$x^*$ is the \textit{conjugate} of complex value $x$,
$\mathds{1}_{i=j}$ is the indicator function equal to one when $i=j$ and equal to zero otherwise,
$const.$ represents a constant independent of the variables of interest,
($\star$) denotes the convolution operation and $(\propto)$ denotes `\textit{proportional to}'.

\begin{figure*}
	\centering
	\includegraphics[width=\textwidth]{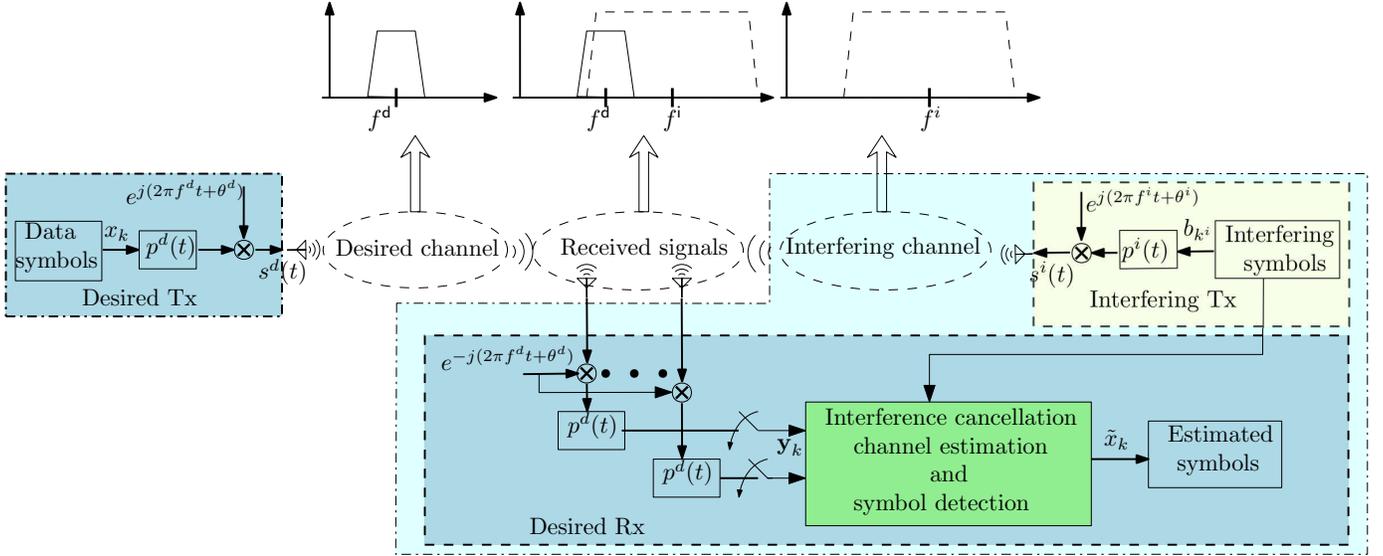}
	\caption{Considered interference scenario}
	\label{fig: system structure}
\end{figure*}

\section{System Model and Problem Statement}

We consider the scenario where two communication links denoted by $\pmb{\mathbb{S}}^{\sf d}$ (desired link) and $\pmb{\mathbb{S}}^{\sf i}$ (interfering link) operate on overlapping frequency bands. The transmitted signal from $\pmb{\mathbb{S}}^{\sf i}$ interferes with the received signal of $\pmb{\mathbb{S}}^{\sf d}$.
One popular assumption usually made in the literature is that interfering and desired signals have identical bandwidths 
where the full-duplex system is a special setting attracting great interests recently.
Our current paper considers the more general scenario in which the frequency bands of the
two communication links can be arbitrarily aligned and their bandwidth ratio is an integer.
The considered setting corresponds to the practical interference scenarios in 
satellite communications \cite{FCC_sat, euro_report}
	and terrestrial communications, e.g., full-duplex relay \cite{ngo2014multipair, huang2017power}.

We further assume that the desired communication channel experiences the fast fading where the channel coefficient changes
from symbol to symbol according to the first order Markov process \cite{wang1996verifying,tan2000first}.
In addition, the interfering channel from the interfering source to the antennas of the desired receiver
is assumed to be line of sight.
In this interference scenario, the involved signals have different bandwidths and
are not synchronized with each other.
This induces a dynamic interference pattern to the desired received signal,
which can be captured by the EICs \cite{Nguyen2017, schwarzenbarth2007}.
We propose to jointly estimate the desired channel coefficients and the EICs with the knowledge of transmitted
symbols from the interfering source and the pilot symbols of the desired signal.

	The considered setting with desired and interfering communications is illustrated in 
	Fig. \ref{fig: system structure}.
	The studied interference scenario occurs in practice when the interfering Tx and the desired Rx are located close to each other
	and the desired Rx has access to the interfering symbols (e.g., via a dedicated connection) as in the full-duplex 
	relay \cite{ngo2014multipair, huang2017power}.
	More details about the system are introduced in the followings.

\subsection{Signal Models}
The transmitted signal of the desired communication with the carrier frequency
$f^{\sf d}$ can be written as
\begin{equation}  \label{eq: desired Tx continuous}
s^{\sf d}(t)=\sum_{k=- \infty}^{\infty} {x_k p^{\sf d}(t{-}kT^{\sf d}{+}\epsilon^{\sf d})e^{j(2\pi f^{\sf d} t +\theta^{\sf d})}}, \\
\end{equation}
where $x_k$ is the $k$th transmitted symbol. 
The pulse shaping function $p^{\sf d}(t)$ has unity gain;
$T^{\sf d}$, $\epsilon^{\sf d}$ and $\theta^{\sf d}$ 
represent the symbol duration, time and phase offsets, respectively.
Similarly, the signal from the interfering source can be written as
\begin{equation} \label{eq: interfering Tx continuous}
s^{\sf i}(t)=\sum_{k^{\sf i}=-\infty}^{\infty} {b_{k^{\sf i}}p^{\sf i}(t{-}k^{\sf i}T^{\sf i}{-}t^{\sf i})e^{j(2\pi f^{\sf i} t+\theta^{\sf i})}}, \\
\end{equation}
where $p^{\sf i}(t)$ denotes the pulse shaping filter with unity gain,
the interfering signal has the center frequency $f^{\sf i}=f^{\sf d}-\Delta f$, 
the $k^{\sf i}$th symbol is $b_{k^{\sf i}}$;
$t^{\sf i}$ and $\theta^{\sf i}$ account for the time/phase difference of the two systems and transmission time delay
from the interfering transmitter to the interfered receiver, respectively.
Assume that there are $N_{\sf r}$ receiver antennas for $\pmb{\mathbb{S}}^{\sf d}$, then the received signal is  
\begin{equation} \label{eq: Rx continuous}
\mathbf{y}(t)=\mathbf{h}^{\sf d}(t)\star s^{\sf d}(t)+\mathbf{h}^{\sf i}(t)\star s^{\sf i}(t){+}\mathbf{w}(t), \\
\end{equation}
where $\mathbf{h}^{\sf d}(t)$ and $\mathbf{h}^{\sf i}(t)$ denote $N_{\sf r} {\times} 1$ vectors of desired and interfering channel impulse responses.

At the receiver of $\pmb{\mathbb{S}}^{\sf d}$, the signals are
down-converted to baseband by using $e^{-j(2\pi f^{\sf d} t+\theta^{\sf d})}$. 
The output signals then pass through a matched filter  having the impulse response $p^{\sf d}(t)$; 
and the filtered continuous signals are sampled at $(kT^{\sf d}+\epsilon^{\sf d})$ to yield the following discrete time signal   
\begin{equation} \label{eq: Rx discrete}
\mathbf{y}_k=\mathbf{h}^{\sf d}_k x_k +\mathcal{I}_k+\mathbf{w}_k, \\
\end{equation}
where $\mathbf{w}_k$ represents the vector of noise having complex Gaussian distribution with covariance matrix $\sigma^2\mathbf{I}_{N_{\sf r}}$
($\mathbf{w}_k$ is called AWGN hereafter);
$\mathcal{I}_k$ denotes the equivalent baseband, discrete time interfering signal
which will be derived shortly.
Firstly, we express  the interference terms in the continuous time domain as follows:   
\begin{equation} \label{eq: I_mt1} 
\mathcal{I}(t)=\left\lbrace \left( \mathbf{h}^{\sf i}(t){\star} s^{\sf i}(t)\right)  e^{-j(2\pi f^{\sf d} t+\theta^{\sf d})} \right\rbrace {\star} p^{\sf d}(t).
\end{equation}

Substituting  $s^{\sf i}(t)$ from \eqref{eq: interfering Tx continuous} into \eqref{eq: I_mt1},
we obtain the equivalent baseband interference signal whose sampled signal
at time  $(kT^{\sf d}+\epsilon^{\sf d})$ is
\begin{equation} \label{eq: Imk}
\mathcal{I}_k =\mathcal{I}(t) |_{t=kT^{\sf d}+\epsilon^{\sf d}} =\mathbf{h}^{\sf i}_k \sum\nolimits_{k^{\sf i}}^{}{b_{k^{\sf i}} c_{k,k^{\sf i}}},
\end{equation}
where $c_{k,k^{\sf i}}$ represents the EIC which is defined in \eqref{eq: EIC general}.
\newcounter{tempequationcounter_1}
\begin{figure*}[t]
	\normalsize
	\setcounter{tempequationcounter_1}{\value{equation}}
	\begin{IEEEeqnarray}{rCl}
		&c_{k,k^{\sf i}}=\int_{-\infty}^{\infty}{p^{\sf d}(kT^{\sf d}+\epsilon^{\sf d}-\tau)p^{\sf i}(\tau-k^{\sf i}T^{\sf i}-t^{\sf i})e^{j\big(2\pi (f^{\sf i}-f^{\sf d})\tau+\theta^{\sf i}+\theta^{\sf d}\big)}d\tau}.
		\label{eq: EIC general}
	\end{IEEEeqnarray}
	\setcounter{equation}{\value{tempequationcounter_1}}
	\hrulefill
\end{figure*}
\addtocounter{equation}{1}

	Suppose that the interfering signal's bandwidth is $\textit{M}$ times larger than
	that of the interfered signal's bandwidth and
	there are $\textit{L}$ symbols of $b_{k^{\sf i}}$'s interfering to each desired symbol $x_k$
	where $L$ should be a multiple of the bandwidth ratio $\textit{M}$ to account for
	the interference in the filter span of the desired signal\footnote{For tractability, 
	the bandwidth ratio $\textit{M}$ is an integer. As a result, the achieved results provide
	performance bounds and approximation for the case where $\textit{M}$ is a real number.}.
	For the considered interference scenario, the bandwidth of the interfering signal is multiple times larger than that of the desired signal.
	Since the bandwidth ratio is an integer, $c_{k,k^{\sf i}}$ in \eqref{eq: EIC general}
	 depends only on the relative difference of $k$, $k^{\sf i}$.
	Hence, for brevity, we denote them as  $ \mathbf{c}=[c_1, c_2,..., c_L]^T$ in the sequel.

\subsection{Channel Models}

The fast fading channel of the desired communication link $\mathbf{h}^{\sf d}_k$ in  \eqref{eq: Rx discrete} is assumed to
follow the first-order Markov model
where the relation of channel coefficients at instants $(k+1)$th and $k$th
can be described as \cite{wang1996verifying}:
\begin{equation} \label{eq: channel_model}
\begin{split}
\mathbf{h}^{\sf d}_{k+1}&= \alpha \mathbf{h}^{\sf d}_k+\sqrt{1-\alpha^2}\boldsymbol{\Delta}_k,\\
\end{split}
\end{equation}
where $\boldsymbol{\Delta}_k$ denotes
a vector of Circular Symmetric Complex Gaussian (CSCG) noise
with zero means and covariance matrix $\sigma^2_{\sf h} \mathbf{I}_{N_{\sf r}}$.
The additive noise term in \eqref{eq: channel_model} is
called channel evolutionary noise
and $\alpha$ is the channel correlation coefficient.
The average Signal to Noise Ratio (SNR) is $\rho=\sigma_{\sf h}^2/\sigma^2$
(called SNR without fading in some previous works \cite{sun2014maximizing}).
Without loss of generality, we let $\sigma^2_{\sf h}=1$.
However, $\sigma^2_{\sf h}$ may appear occasionally in several expressions whenever needed.

	The Markovian channel model can accurately capture the practical Clarke channel model,
	which has been validated in \cite{wang1996verifying, tan2000first}. Moreover, the Markovian channel model
	has been widely adopted in the literature \cite{sadeghi2008finite, zhu2009message, zhang2013blind, 
	sun2014maximizing, tan2000first, mustafa2015separation, hu2018secure}.
	In fact, the authors of \cite{zhu2009message} have conducted the model mismatching study,
	where the actual channel follows the Clarke model and the assumed channel is the Gaussian-Markov model,
	and they have found that the mismatch is negligible.

We assume that the receiver has perfect 
information about the interfering channel gains $\mathbf{h}_k^{\sf i}$ which
correspond to the line of sight link as assumed.
Therefore, the interfering channel gains vary slowly over time and they can be estimated accurately.

\subsection{Problem Statement} 

Using the result of $\mathcal{I}_k$ in \eqref{eq: Imk}, we can rewrite the received signal in (\ref{eq: Rx discrete}) as
\begin{equation} \label{eq: Rx final}
\begin{split}
\mathbf{y}_k
& =\mathbf{h}^{\sf d}_kx_k + \sum_{l=1}^{L}{\left( \mathbf{h}^{\sf i}_k b_{Mk+l}\right)  c_l }+\mathbf{w}_k\\
& =\mathbf{h}^{\sf d}_kx_k + \sum_{l=1}^{L}{ \mathbf{b}_{k,l}  c_l }+\mathbf{w}_k,
\end{split}
\end{equation}
where $\mathbf{b}_{k,l}=\mathbf{h}^{\sf i}_k b_{Mk+l}$.
Then, we can rewrite \eqref{eq: Rx final} in a matrix form as follows:
\begin{equation}
\mathbf{y}_k=\mathbf{h}^{\sf d}_kx_k + \mathbf{B}_k\mathbf{c}+\mathbf{w}_k,
\end{equation} 
where
	$\mathbf{B}_k$ is the $N_{\sf r} \times L$ matrix whose $l$th column is $\mathbf{b}_{k,l}$.
	We will call $\mathbf{B}_k$ the interference matrix hereafter.
	Recall that the interfering symbols $\mathbf{b}_{Mk+l}$ and the interfering channel gains $\mathbf{h}^{\sf i}_k$ are assumed to be known. 
	Therefore, $\mathbf{B}_k$ is known by the desired receiver.

In this paper,
$\mathbf{y}_k$ is referred to as the \textit{received signal} or \textit{observation} interchangeably.
Since the interfering channels are known and captured
in the interference matrix $\mathbf{B}_k$, we will omit the superscript $\sf d$ in the desired channel notation, i.e., $\mathbf{h}_k^{\sf d}$ becomes $\mathbf{h}_k$. From now on, \textit{channels} means desired channels discussed in the previous sections.

This paper aims to address the following questions:
\begin{itemize}
	\item [1)] Given the interference matrix $\mathbf{B}_k$,
	the observations $\mathbf{y}_k$ and the pilot symbols, 
	how can one cancel the interference and detect data symbols reliably?
	\item[2)] What are the effects of fast fading channel evolutionary noise to the overall system performances (EIC estimation,
	interference cancellation, channel estimation, and symbol detection)?
	\item[3)] Is there an optimal frame design (i.e., optimal pilot density) that maximizes the throughput in the presence of fast fading and interference?
\end{itemize}

In the next sections, we will provide the answers for these questions.

\section{Proposed Channel Estimations and Interference Cancellation Strategy}

Even though the MMSE method has been widely used in channel estimation, this method
relies heavily on the knowledge of the time-domain  channel correlation \cite{ma2003optimal, feng2010estimation, zemen2012adaptive, xu2019adaptive}.
In the presence of interference, MMSE can only be applied after the interference is canceled out.
Moreover, its achieved performance depends on the interference cancellation techniques and the resulted residual interference.
In addition, MMSE estimators typically require matrix inversion with complexity scaling with the number of pilot symbols,
which may become unaffordable for long frames.
These drawbacks of the MMSE method
motivate us to use the MAP estimator instead where
the MAP estimator can be used to estimate the channel coefficients.
Furthermore, the MAP estimator is usually preferable to other estimation techniques regarding both bias and variance
for the setting with a small number of observations,
which corresponds to the small number of pilot symbols in our considered frame \cite{yeredor2000joint}.

In this section, we propose a two-phase design framework for estimation of the EICs and symbol detection.
In the first phase, the EICs are estimated at each pilot position using the maximum likelihood (ML) approach.
Then, we take the average of the estimates of $\mathbf{c}$ over all pilot positions
to obtain a reduced-variance estimate of $\mathbf{c}$ compared to its estimates at different pilot positions.
After that the interference is subtracted from the received signal and the channel coefficients are estimated at pilot positions.
In the second phase, the  \textit{a posteriori probability} of data symbols is derived, given the estimated channel coefficients at 
the pilot positions before and after the data intervals, then the data symbols are detected based on that probability.
Fig. \ref{fig:Block} illustrates our proposed design for one particular frame. 
\begin{figure}[ht]
	\centering
	\includegraphics[scale=0.72]{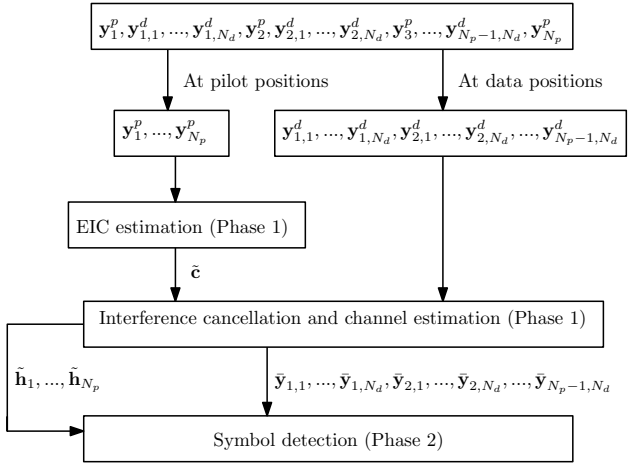}
	\caption{Illustration of the proposed design}
	\label{fig:Block}
\end{figure}

Channel estimation and symbol detection are performed in each frame.
We consider the scattered pilot frame structure in the time domain with $N_{\sf d}$ data symbols between two consecutive pilot symbols,
and there are $N_{\sf p}$ pilot symbols in a frame \cite{tong2004pilot, asyhari2017orthogonal}.
Typical symbol arrangement in a frame is expressed as
$[x^{\sf p}_1,x^{\sf d}_{1,1}{,}...{,}x^{\sf d}_{1,N_{\sf d}}{,}x^{\sf p}_{2},x^{\sf d}_{1,2},...,x^{\sf d}_{2,N_{\sf d}}{,}...{,}x^{\sf d}_{N_{\sf p}{-}1,N_{\sf d}},x^{\sf p}_{N_{\sf p}}]$,
where $x^{\sf p}_{i}$ denotes the $i$th pilot symbol, and $[x^{\sf d}_{1,i},...,x^{\sf d}_{i,N_{\sf d}}]$ denotes data symbols between the $i$th and $(i{+}1)$th pilot symbols. Fig. \ref{fig: pilot} illustrates this pilot arrangement.
\begin{figure}[ht]
	\centering
	\includegraphics[scale=0.83]{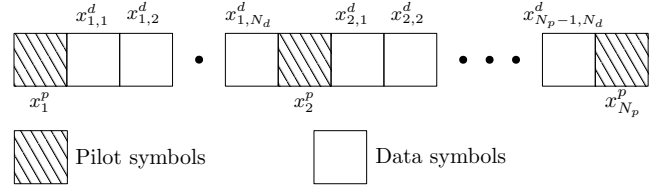}
	\caption{Pilot and data symbol arrangement in a frame}
	\label{fig: pilot}
\end{figure}

\subsection{Phase 1: Estimation of Interference and Channel Coefficients}

In the first phase, we are interested in estimating $\mathbf{c}$
and $\mathbf{h}_n^{\sf p}, n=1,...,N_{\sf p}$ given the observations $\mathbf{y}^{\sf p}_{1:N_{\sf p}}$. 
For brevity, the superscript $\sf p$ is omitted in this section,
i.e., $x_i^{\sf p}$ becomes $x_i$.
We denote $\mathbf{Y}=[\mathbf{y}_{1:n-1}, \mathbf{y}_n, \mathbf{y}_{n+1:N_{\sf p}}]$.
We have the knowledge of the distribution of $\mathbf{h}_n$, so we use the MAP criteria to estimate $\mathbf{h}_n$.
Note that either $p(\mathbf{h}_n|\mathbf{Y})$ or $p(\mathbf{h}_n,\mathbf{Y})$ can be used,
since $p(\mathbf{h}_n,\mathbf{Y})=p(\mathbf{h}_n|\mathbf{Y})p(\mathbf{Y})$
and $p(\mathbf{Y})$ is independent of the parameter of interest $\mathbf{h}_n$.
Recall also that the EICs $\mathbf{c}$ are unknown, deterministic parameters within a frame.
Therefore, the joint estimation 
criteria for $\mathbf{c}$ and $\mathbf{h}_n$ can be expressed as
\begin{equation} \label{eq: h estimation criteria}
\begin{split}
\left\lbrace \tilde{\mathbf{c}}_n, \tilde{\mathbf{h}}_n\right\rbrace
&= \text{argmax} \:  p(\mathbf{h}_n,\mathbf{Y}|\mathbf{c}).
\end{split}
\end{equation}

For notational convenience, we omit $\mathbf{c}$ in the following distributions,
when there is no confusion,
i.e., $p(\mathbf{h}_n,\mathbf{Y}|\mathbf{c})$ is simply written as $p(\mathbf{h}_n,\mathbf{Y})$.
In order to estimate $\mathbf{h}_n$ and $\mathbf{c}$ according to \eqref{eq: h estimation criteria},
we need to find $p(\mathbf{h}_n,\mathbf{Y})$. 
Therefore, we provide the following theorem which states the log likelihood of the received signals and the channel coefficients at pilot positions.
\begin{theorem} \label{Theorem: prob of h,Y}
	The log likelihood of the received signals and channel coefficients at pilot position $n$ is
	\begin{equation} \label{eq: LLH h_n Y}
	\begin{split}
	\mathcal{L}_{\mathbf{h}_n,\mathbf{Y}} &=\log(p(\mathbf{h}_n,\mathbf{Y}) )\\
	{=}{-}&
	\sum_{i=1}^{N_{\sf p}}{\left( \mathbf{y}_{i}{-}\boldsymbol{\mu}_{i,n}\right) ^H \boldsymbol{\Sigma}_{i,n}^{-1}\left( \mathbf{y}_{i}{-}\boldsymbol{\mu}_{i,n}\right)}
	{-}\mathbf{h}_n^H\mathbf{h}_n {+} const. .\\
	\end{split}
	\end{equation}
\end{theorem}
\begin{proof}
	The derivation and related parameters ($\boldsymbol{\mu}_{i,n}, \boldsymbol{\Sigma}_{i,n}$) can be found in Appendix A.
\end{proof}

We estimate the desired channel and EIC by maximizing $\mathcal{L}_{\mathbf{h}_n,\mathbf{Y}}$.
As can be shown in the derivation later, the exponent of $p(\mathbf{h}_n,\mathbf{Y}|\mathbf{c})$ can be decomposed into two quadratic terms
where one term contains $\mathbf{h}_n$
and the other contains only $\mathbf{c}$ and not $\mathbf{h}_n$.
Since there are two variables to be optimized (i.e., $\mathbf{h}_n$ and $\mathbf{c}$), we first derive 
the optimal $\mathbf{h}_n$ with respect to  $\mathbf{c}$ then we derive
the optimal $\mathbf{c}$ by maximizing the corresponding objective function achieved with the optimal $\mathbf{h}_n$. 

$\bigstar$ \textit{Step 1}: Derivation of the optimal $\mathbf{h}_n$ for a given $\mathbf{c}$

The sum of quadratic terms in (\ref{eq: LLH h_n Y}) can be re-written as
\begin{equation} \label{eq: LLH hn Y -const}
\begin{split}
\tilde{\mathcal{L}}_{\mathbf{h}_n,\mathbf{Y}}&{=}-\mathbf{h}_n^H\mathbf{h}_n \\
{-}&\sum_{i=1}^{N_{\sf p}}{\!
	\left( \mathbf{y}_{i,n}{-}x_{i,n}\mathbf{h}_n{-}\mathbf{B}_{i,n}\mathbf{c}\right)^{\!H}
	\! \boldsymbol{\Sigma}_{i,n}^{-1}\!
	\left( \mathbf{y}_{i,n}{-}x_{i,n}\mathbf{h}_n{-}\mathbf{B}_{i,n}\mathbf{c}\right)}\\
&{=}{-}
(\mathbf{h}_n-\tilde{\mathbf{h}}_n)^H
\mathbf{A}_n
(\mathbf{h}_n-\tilde{\mathbf{h}}_n)
-\mathcal{C}_n,
\end{split}
\end{equation} 
where we omit the constant in (\ref{eq: LLH h_n Y}).
$\mathbf{A}_n, \tilde{\mathbf{h}}_n$ and $\mathcal{C}_n$ are defined as
\begin{equation}\label{eq: step1_para}
\begin{split}
\mathbf{A}_n&{=}\mathbf{I}_{N_{\sf r}}+\sum_{i=1}^{N_{\sf p}}{\omega_{i,n}^2\boldsymbol{\Sigma}_{i,n}^{-1}},\\
\tilde{\mathbf{h}}_n&{=}\mathbf{A}_n^{-1}\left( \sum_{i=1}^{N_{\sf p}}{x_{i,n}^*\boldsymbol{\Sigma}_{i,n}^{-1}\left( \mathbf{y}_{i,n}-\mathbf{B}_{i,n}\mathbf{c}\right)} \right), \\
\mathcal{C}_n&{=}{-}\tilde{\mathbf{h}}_n^H\mathbf{A}_n\tilde{\mathbf{h}}_n{+} \!\!
\sum_{i=1}^{N_{\sf p}}{
	\!\left(
	\mathbf{y}_{i,n}\!{-}\mathbf{B}_{i,n}\mathbf{c}
	\right)^H \! 
	\boldsymbol{\Sigma}^{-1}_{i,n}
	\left(
	\mathbf{y}_{i,n}{-}\mathbf{B}_{i,n}\mathbf{c}
	\right)},\\
\end{split}
\end{equation}
where $\omega_{i,n}, x_{i,n}, \mathbf{y}_{i,n}, \mathbf{B}_{i,n}$ and the related parameters are defined in (\ref{eq: xyB})-(\ref{eq: para1}).
For notational simplicity, we denote the `sign indicator' $j_{i,n}=-1$ for $i>n$, $j_{i,n}=1$ for $i<n$ and $j_{i,n}=0$ for $i=n$.
Since $\mathbf{A}_n$ is positive definite, the optimal $\mathbf{h}_n$ that maximizes $\tilde{\mathcal{L}}_{\mathbf{h}_n,\mathbf{Y}}$ in (\ref{eq: LLH hn Y -const}) is $\tilde{\mathbf{h}}_n$.
Note that, when the desired channels are independent, we have $\mathbf{A}_n=a_n\mathbf{I}_{N_{\sf r}}$, where
\begin{equation}
a_n=1+\sum_{i=1}^{N_{\sf p}}{\frac{\omega_{i,n}^2}{\sigma_{i,n}^2}}.
\end{equation}

\newcounter{tempequationcounter_2}
\begin{figure*}[ht]
	\normalsize
	\setcounter{tempequationcounter_1}{\value{equation}}
	\begin{IEEEeqnarray}{lCl}
		x_{i,n}=\omega_{i,n}x_i, \quad
		\mathbf{y}_{i,n}=
		\mathbf{y}_i{-}\beta_{i,n}\mathbf{y}_{i+j_{i,n}}, \quad
		\mathbf{B}_{i,n}=\mathbf{B}_i{-}\beta_{i,n}\mathbf{B}_{i+j_{i,n}}. \label{eq: xyB} \\
		\omega_{i,n}=\left\{
		\begin{array}{ll}
			\frac{\alpha_{\sf p}^{|n-i|}}{1+\rho(1-\alpha_{\sf p}^{2(|n-i|-1)})}, &i \neq n\\
			1, &i=n
		\end{array}
		\right.
		, \quad \quad \quad \quad \quad
		\beta_{i,n}=
		\left\{
		\begin{array}{ll}
			\frac{x_ix_{i+j_{i,n}}^*\rho\alpha_{\sf p}\left(1- \alpha_{\sf p}^{2(|n-i|-1)}\right) }
			{1+\rho \left(1- \alpha_{\sf p}^{2(|n-i|-1)}\right) }, &i \neq n\\
			0, &i=n
		\end{array}
		\right.
		. \label{eq: para1}\\
		\mathbf{D}_n=\sum \nolimits _{i=1}^{N_{\sf p}}{\mathbf{B}_{i,n}^H\boldsymbol{\Sigma}_{i,n}^{-1}\mathbf{B}_{i,n}}
		{-}\left( \sum \nolimits _{i=1}^{N_{\sf p}}{x_{i,n}^*\boldsymbol{\Sigma}_{i,n}^{-1}\mathbf{B}_{i,n}}\right)^H\mathbf{A}_n^{-1}
		\left( \sum \nolimits _{i=1}^{N_{\sf p}}{x_{i,n}^*\boldsymbol{\Sigma}_{i,n}^{-1}\mathbf{B}_{i,n}}\right). \label{eq: D_tilde}\\ 
		\tilde{\mathbf{c}}_n=\mathbf{D}^{-1}_n \left\lbrace 
		\sum \nolimits _{i=1}^{N_{\sf p}}{\mathbf{B}_{i,n}^H\boldsymbol{\Sigma}_{i,n}^{-1}\mathbf{y}_{i,n}}
		- \left( \sum \nolimits _{i=1}^{N_{\sf p}}{x_{i,n}^*\boldsymbol{\Sigma}_{i,n}^{-1}\mathbf{B}_{i,n}}\right)^H\mathbf{A}_n^{-1}
		\left( \sum \nolimits _{i=1}^{N_{\sf p}}{x_{i,n}^*\boldsymbol{\Sigma}^{-1}_{i,n}\mathbf{y}_{i,n}}\right)
		\right\rbrace. \label{eq: c_tilde}\\
		\mathcal{F}=\sum \nolimits _{i=1}^{N_{\sf d}}{\left[ \left(\tau_2 \boldsymbol{\Gamma}_{i,1}\mathbf{h}_h{+}\mathds{1}_{i=N_{\sf d}}\tau_2\mathbf{h}_t {+}\sum \nolimits _{j=1}^{i}{\frac{x_j^*}{\sigma^2}\boldsymbol{\Gamma}_{i,j}\mathbf{y}_j}\right)^H \mathbf{S}_i\right. } 
		\left. 	\left( \tau_2 \boldsymbol{\Gamma}_{i,1}\mathbf{h}_h{+}\mathds{1}_{i=N_{\sf d}}\tau_2\mathbf{h}_t{+}\sum \nolimits _{j=1}^{i}{\frac{x_j^*}{\sigma^2}\boldsymbol{\Gamma}_{i,j}\mathbf{y}_j}\right) \right]. \label{eq: F}\\
		\mathbf{S}_i^{-1}=
		\left[ \frac{1}{\sigma^2}+(1+\alpha^2)\tau_1 \right]\mathbf{I}_{N_{\sf r}}
		-\mathds{1}_{i>1}\tau_2^2\mathbf{S}_{i-1}, \quad
		\bar{\mathbf{h}}_i=\left\{
		\begin{array}{ll}
			\mathbf{S}_i\left[ \tau_2\boldsymbol{\Gamma}_{i,1}\mathbf{h}_h{+} \tau_2\mathbf{h}_{i+1}{+}\sum_{j=1}^{i}{ \frac{x_j^*}{\sigma^2}\boldsymbol{\Gamma}_{i,j}\mathbf{y}_j} \right] , &i{<}N_{\sf d}\\[5pt]
			\mathbf{S}_i\left[ \tau_2\boldsymbol{\Gamma}_{i,1}\mathbf{h}_h{+} \tau_2\mathbf{h}_{t}{+}\sum_{j=1}^{i}{ \frac{x_j^*}{\sigma^2}\boldsymbol{\Gamma}_{i,j}\mathbf{y}_j} \right] , &i=N_{\sf d}
		\end{array}. \label{eq: para3}
		\right.
	\end{IEEEeqnarray} 
	\setcounter{equation}{\value{tempequationcounter_1}}
	\hrulefill
\end{figure*}
\addtocounter{equation}{6}

$\bigstar$ \textit{Step 2}: Derivation of the optimal $\mathbf{c}$

When $\mathbf{h}_n=\tilde{\mathbf{h}}_n$, the function in (\ref{eq: LLH hn Y -const}) is equal to $-\mathcal{C}_n$ which only depends on $\mathbf{c}$ where
\begin{equation}
\begin{split}
\mathcal{C}_n=&\sum \nolimits_{i=1}^{N_{\sf p}}{\left( \mathbf{y}_{i,n}{-}\mathbf{B}_{i,n}\mathbf{c}\right)^H\boldsymbol{\Sigma}^{-1}_{i,n}\left( \mathbf{y}_{i,n}{-}\mathbf{B}_{i,n}\mathbf{c}\right)}\\
{-}&\left\lbrace \left( \sum\nolimits_{i=1}^{N_{\sf p}}{x_{i,n}^*\boldsymbol{\Sigma}_{i,n}^{-1}\left( \mathbf{y}_{i,n}-\mathbf{B}_{i,n}\mathbf{c}\right)} \right)^H\mathbf{A}_n^{-1} \right.\\
&\quad\left. \left( \sum\nolimits_{i=1}^{N_{\sf p}}{x_{i,n}^*\boldsymbol{\Sigma}_{i,n}^{-1}\left( \mathbf{y}_{i,n}-\mathbf{B}_{i,n}\mathbf{c}\right)} \right)\right\rbrace \\
=& (\mathbf{c}-\tilde{\mathbf{c}}_n)^H\mathbf{D}_n(\mathbf{c}-\tilde{\mathbf{h}}_n)+const.,
\end{split}
\end{equation}
where $\mathbf{D}_n$ and $\tilde{\mathbf{c}}_n$ are defined in (\ref{eq: D_tilde}),(\ref{eq: c_tilde}).  
It can be verified that $\mathbf{D}_n$ is positive definite by using the \textit{Cauchy}-\textit{Schwarz} inequality.
The proof of this property can be found in Appendix B.
Therefore, the optimal $\mathbf{c}$ that maximizes $\tilde{\mathcal{L}}_{\mathbf{h}_n,\mathbf{Y}}$ in (\ref{eq: LLH hn Y -const}) is $\tilde{\mathbf{c}}_n$.
We take the average over all $\tilde{\mathbf{c}}_n, n=1,...,N_{\sf p}$ to yield a reduced-variance estimate of $\mathbf{c}$.
Consequently, the resulting estimated EIC vector can be written as  
\begin{equation}\label{eq: universal estimation of c}
\tilde{\mathbf{c}}=\frac{1}{N_{\sf p}}\sum_{n=1}^{N_{\sf p}}{\tilde{\mathbf{c}}_n}.
\end{equation}

The joint interference estimation, cancellation and channel estimation algorithm is described in Algorithm 1.
\begin{algorithm}[ht]
	\caption{Estimation of EICs, Desired Channel Coefficients, and Interference Cancellation}
	\begin{algorithmic}[1]
		\FOR{$n=1:N_{\sf p}$}
		{
			\FOR{$i=1:N_{\sf p}$}
			{
				\STATE Compute $x_{i,n},\mathbf{y}_{i,n},\mathbf{B}_{i,n},\boldsymbol{\Sigma}_{i,n}$ in (\ref{eq: step1_para}),(\ref{eq: xyB}).
			}
			\ENDFOR
			\STATE Compute 	$\mathbf{A}_n,\mathbf{D}_n$, and then $\tilde{\mathbf{c}}_n$ in (\ref{eq: step1_para}), (\ref{eq: D_tilde}),(\ref{eq: c_tilde}).
		} 
		\ENDFOR
		\STATE Compute $\tilde{\mathbf{c}}$ in (\ref{eq: universal estimation of c}) and subtract the interference.
		\FOR {$n=1:N_{\sf p}$}
		{
			\STATE Estimate $\mathbf{h}_n$ as $\tilde{\mathbf{h}}_n$ in (\ref{eq: step1_para}).
		}
		\ENDFOR
		\STATE End of algorithm.
	\end{algorithmic}
\end{algorithm}

\subsection{Symbol Detection}

With the estimated $\tilde{\mathbf{c}}$, we can subtract the interference,
and
the channel coefficients at pilot positions are estimated as $\tilde{\mathbf{h}}_n$ given in
(\ref{eq: step1_para}) with $\mathbf{c}$ substituted by $\tilde{\mathbf{c}}$ 
in (\ref{eq: universal estimation of c}).
The estimated channel coefficients at pilot positions will be used for the
symbol detection as described in the following.

We will describe the symbol detection for the interval $x_i^{\sf p}x_{i,1}^{\sf d}x_{i,2}^{\sf d}...x_{i,N_{\sf d}}^{\sf d}x_{i+1}^{\sf p}$. 
The method can be applied and repeated for other intervals.
For simplicity, we omit the pilot index $i$ and superscript $(\sf d)$ in this section,
i.e., the channel coefficients are denoted as $[\mathbf{h}_h,\mathbf{h}_{1:N_{\sf d}},\mathbf{h}_t]$,
where $\mathbf{h}_h$ represents the known channel coefficient at the pilot symbol right before the considered interval
and $\mathbf{h}_t$ represents known channel coefficient at the pilot symbol right after the considered interval.

In \cite{mahamadu2018fundamental},
the optimum diversity detection is derived to
detect symbols individually based on the interpolated channel coefficients at the corresponding positions
in the interference-free scenario.
This method, however, requires expensive matrix inversion because the matrix size corresponds to the number of pilot symbols.
Alternatively, we provide two different symbol detection methods where
the first method is based on series symbol detection
which will be shown to outperform the optimum individual detector ODD
at the cost of high complexity,
while the second method achieves very close (almost identical) SER to that due to the ODD but with significantly lower complexity.
These detection methods are described in the following.

\subsubsection{Series Symbol MAP Detection (S-MAP)}
The symbols in an interval are detected as
\begin{equation} \label{eq: Series MAP detection}
\begin{split}
\tilde{\mathbf{x}}_{1:N_{\sf d}}&=\text{argmax} \quad  p\left(\mathbf{x}_{1:N_{\sf d}} |\mathbf{h}_h,\mathbf{h}_t,\mathbf{y}_{1:N_{\sf d}} \right).
\end{split}
\end{equation}

We now characterize the log likelihood function in the following theorem. 
\begin{theorem}
	The log likelihood of data symbols conditioned on the received signals and the channel coefficients at pilot positions right after and before
	the interval can be expressed in a sum of quadratic functions of data symbols $\mathbf{x}$ as 
	\begin{equation}
	\log\left(  p\left(\mathbf{x}_{1:N_{\sf d}} |\mathbf{h}_h,\mathbf{h}_t,\mathbf{y}_{1:N_{\sf d}} \right)\right)=\mathcal{F}+const.,
	\end{equation}
\end{theorem}
where $\mathcal{F}$ and the related parameters can be found in (\ref{eq: F}), (\ref{eq: para3}) and Appendix C.
\begin{proof}
	The proof and related parameters can be found in Appendix C.
\end{proof}

By enumerating all possible vectors $\mathbf{x}=[x_1,...,x_{N_{\sf d}}]$
from the constellation points and calculating the corresponding $p\left(\mathbf{x}_{1:N_{\sf d}} |\mathbf{h}_h,\mathbf{h}_t,\mathbf{y}_{1:N_{\sf d}} \right)$, we
are able to obtain the optimally detected symbols by \eqref{eq: Series MAP detection}.

\subsubsection{Individual Symbol MAP Detection (I-MAP)}

The individual symbol detection method presented in \cite{nguyen2018channel}
determines the detected symbol $x_i$ based on (\ref{eq: Series MAP detection}).
However, because $\tilde{x}_i$ is computed from $\tilde{x}_j, j<i$,
this method suffers from error propagation,
which increases the error rates of symbols in the middle of the interval.
To address this limitation, we propose to estimate $x_i$ individually as 
\begin{equation} \label{eq: symbol detection new}
\begin{split}
\tilde{x}_{i}&=\text{argmax} \quad  p\left(x_{i} |\mathbf{h}_h,\mathbf{h}_t,\mathbf{y}_{i} \right).
\end{split}
\end{equation}
Using similar derivations as those used to obtain the results in Theorem 2,
we have\footnote{Upon deriving $\breve{\mathbf{h}}_i$, the normalized technique employed is similar to that employed in the well-known Maximal Ratio Combining technique.}
\begin{equation}\label{eq: x_detection}
\begin{split}
\tilde{x}_i&=\frac{\breve{\mathbf{h}}_i^H\mathbf{y}_i}{\|\breve{\mathbf{h}}_i^H\mathbf{y}_i\|}, i=1,...,N_{\sf d},\\
\breve{\mathbf{h}}_i&=\frac{\alpha^i}{1-\alpha^{2i}}\mathbf{h}_h+\frac{\alpha^{N_{\sf d}+1-i}}{1-\alpha^{2(N_{\sf d}+1-i)}}\mathbf{h}_t.
\end{split}
\end{equation}

Then, the detected symbols can be found by mapping $\tilde{x}_i$ to the closest point on the constellation.
This method does not suffer from error propagation and its achievable performance is
less sensitive to the positions $i$ of the data symbol in each detection interval.
We summarize the proposed joint channel estimation and symbol detection in Algorithm 2.
\begin{algorithm}[t]
	\caption{Symbol Detection Over Fast Fading Channel (I-MAP)}
	\begin{algorithmic}[1]
		\FOR{$n=1:N_{\sf p}$}
		{
			\FOR{$i=1:N_{\sf d}$}
			{
				\STATE Estimate $\tilde{x}^{\sf d}_{i,n}$ from (\ref{eq: x_detection}) and assign $\tilde{x}^{\sf d}_{i,n}$ to the closest point in the constellation.
			}
			\ENDFOR
		}
		\ENDFOR
		\STATE End of algorithm.
	\end{algorithmic}  
\end{algorithm}   

\subsection{Iterative Algorithm for Interference Cancellation, Channel Estimation and Symbol Detection}
In practice, the joint channel estimation, interference cancellation, and data detection are often performed iteratively 
\cite{shi2017frequency}. Moreover, if the data detection is sufficiently reliable,
detected data symbols can act as pilot symbols to support the interference cancellation and channel estimation,
which can potentially improve the detection performance.
In this section, we propose an iterative approach for
interference cancellation, channel estimation, and symbol detection based on the previous two-phase method.
For convenience purposes, we now denote the desired symbols in the frame as
$x_n, n=1,...,(N_{\sf p}-1)(N_{\sf d}+1)+1$, where $x_n, n=1, 1+N_{\sf d}+1, 1+2(N_{\sf d}+1), ...$ are pilot symbols in the previous notations.

\subsubsection{Interference Cancellation and Channel Estimation}
Since all symbols $x_n$ are known (at pilot positions) or detected (at data positions),
they are all treated as pilot symbols.
Therefore, the number of newly considered pilot symbols is now $\hat{N}_{\sf p}=(N_{\sf d}+1)(N_{\sf p}-1)+1$ (symbols in 
the whole frame) and the correlation coefficient of channel gains at two consecutive pilot positions is $\hat{\alpha}_{\sf p}=\alpha$ (instead of $\alpha^{N_{\sf d}+1}$).
The interference estimation, interference cancellation, and channel estimation are performed
as presented in Section III.A.

\subsubsection{Symbol Detection}
Let the estimated channel gains at position $n$ be $\breve{\mathbf{h}}_{n}$.
In order to detect the symbol $x_n$,
we now use the knowledge of $\breve{\mathbf{h}}_{n+1}$ and $\breve{\mathbf{h}}_{n-1}$ as if $n+1$ and $n-1$ are two pilot positions.
Apply the I-MAP technique\footnote{Now as there is only one data symbol between two pilot symbols, S-MAP and I-MAP produce identical results.}
in \eqref{eq: x_detection}, we have
\begin{equation}\label{eq: x_detection_ba}
\begin{split}
\tilde{x}_n&=\frac{\hat{\mathbf{h}}_n^H\mathbf{y}_n}{\|\hat{\mathbf{h}}_n^H\mathbf{y}_n\|}, n=2,...,(N_{\sf p}-1)(N_{\sf d}+1),\\
\hat{\mathbf{h}}_i&=\frac{\alpha}{1-\alpha^{2}}\left(\breve{\mathbf{h}}_{n-1} +\breve{\mathbf{h}}_{n+1} \right).
\end{split}
\end{equation}
After $\tilde{x}_n$ are detected, in the next iterations,
interference cancellation, channel estimation and data detection are performed until convergence is reached.
The algorithm converges when there is no change in the detected
data symbols.
Though the convergence guarantee is difficult to prove, simulation results show that the convergence is achieved after only a few iterations.
We summarize this iterative approach in Algorithm 3.

\begin{algorithm}[t]
	\caption{Iterative Algorithm for Channel Estimation, Interference Cancellation and Data Detection}
	\begin{algorithmic}[1]
		\STATE Perform Algorithm 1 for interference cancellation and channel estimation.
		\STATE Perform Algorithm 2 for I-MAP symbol detection.
		\WHILE{(true)}
		{
			\STATE Perform Algorithm 1 for interference cancellation and channel estimation with $\hat{\alpha}_{\sf p}=\alpha$ and $\hat{N}_{\sf p}=(N_{\sf d}+1)(N_{\sf p}-1)+1$.
			\STATE Perform Algorithm 2 for I-MAP symbol detection with $\hat{N}_{\sf d}=1$. The detected data symbols are denoted as $\bar{\mathbf{x}}^{\sf i}$.
			\IF {$\bar{\mathbf{x}}^{\sf i}$==$\bar{\mathbf{x}}^{\sf (i-1)}$}
			{
				\STATE Break the loop (Convergence is reached).	
			}
			\ELSE
			{
				\STATE Increase ${\sf i}$ and go to the next iteration.
			}
			\ENDIF
		}
		\ENDWHILE
		\STATE End of algorithm.
	\end{algorithmic}  
\end{algorithm}

\section{Performance Analysis}

In this section, we conduct performance analysis for the proposed design framework in Sections III.A and 
III.B\footnote{Due to the stochastic nature of the channel model and the design, analysis 
of the iterative algorithm is very involved, which is beyond the scope of this paper.
Nevertheless, the analysis of the proposed non-iterative two-phase design provides many insights that help
 explain the behaviors of the iterative algorithm.
In-depth analysis of the iterative algorithm is left for our future works.}.
For benchmarking, we first consider the interference-free scenario and 
inspect the effects of AWGN and channel evolutionary noise to the residual interference $\boldsymbol{\nu}_n$.
As shown from the analysis later, the mean square of the channel estimation error (CEE) in the interference-free
scenario approaches zero as the SNR tends to infinity.
In the considered interference scenario,
we prove that the residual interference and the channel estimation error are independent of the interfering power.
Finally, based on the analysis of the estimation error, we demonstrate how the actual residual interference affects the symbol detection
and derive the achievable SER.

In the following analysis, we investigate the channel estimation error (CEE, denoted as $\boldsymbol{\nu}_n$) 
and residual interference (denoted as $\boldsymbol{\upsilon}_n$) which are defined as follows:
\begin{equation}
\begin{split}
\boldsymbol{\nu}_n&= \mathbf{h}_n-\tilde{\mathbf{h}}_n,\\
\boldsymbol{\upsilon}_n&= \mathbf{B}_n\left( \mathbf{c}-\tilde{\mathbf{c}} \right). \\
\end{split}
\end{equation}
\subsection{Channel Estimation in Interference-free Scenario}
In the  interference-free case, the estimate of $\mathbf{h}_n$ is
\begin{equation} \label{eq: h_estimate}
\tilde{\mathbf{h}}_n=\mathbf{A}_n^{-1}\left( \sum _{i=1}^{ N}{x_{i,n}^*\boldsymbol{\Sigma}_{i,n}^{-1} \mathbf{y}_{i,n}} \right).
\end{equation}

We characterize the performance of this channel estimator in the following proposition\footnote{The fact that the 
effect of channel evolutionary noise diminishes as SNR goes to infinity suggests that the error floor in channel 
estimation reported in \cite{nguyen2018channel} comes from the residual interference. The later analysis will confirm this prediction.}.

\begin{proposition}
	The channel estimation error $\boldsymbol{\nu}_n$ has Gaussian distribution with zero mean.
	Moreover, the effect of channel evolutionary noise to the channel estimation error is negligible as the SNR tends to infinity.
\end{proposition}
\begin{proof}
	Please see Appendix D.
\end{proof}
\subsection{Residual Interference Analysis}
For the derived estimators for $\mathbf{c}$ and $\mathbf{h}_n$ under the considered interference scenario,  
the resulting residual interference is characterized in the following propositions.

\begin{proposition}
	The EIC estimation is unbiased and the residual interference follows the Gaussian distribution with zero mean. 
	Moreover, the residual interference is independent of $\mathbf{c}$ and has bounded power as the interference power goes to infinity.
\end{proposition} 
\begin{proof}
	Please see Appendix E.
\end{proof}

\begin{proposition}
	There is a floor for the residual interference power, i.e., as $\rho$ goes to infinity, the residual interference power approaches $\tilde{\sigma}^2_{\sf i}=\frac{\alpha_{\sf p}^2(1-\alpha_{\sf p}^2)}{N_{\sf p}}$.
\end{proposition} 
\begin{proof}
	Please see Appendix F.
\end{proof}

The channel estimation is performed based on the observations after interference cancellation.
Therefore, the floor of residual interference corresponds to the floor in channel estimation performance.
This also means that the achieved SINR after cancellation is bounded. This result is stated
in the following proposition.

\begin{proposition}
	As the SNR goes to infinity, the SINR after interference cancellation\footnote{Since the interference is efficiently canceled, 
		the probably most important parameter before interference cancellation is the SNR; therefore, we use the term "SNR before cancellation" but not "SINR before cancellation" to reflect this. After interference cancellation, the residual interference  is 
		irreducible and affects directly the performance of the detection process; hence, the term "SINR after cancellation" is used.} approaches
	$\tilde{\rho}=\frac{N_{\sf p}}{\alpha_{\sf p}^2(1-\alpha_{\sf p}^2)}$.
\end{proposition} 
\begin{proof}
	After interference cancellation, the achievable SINR is affected by the channel estimation error and the residual interference. 
	According to Proposition 1, the channel estimation error vanishes as $\rho \to \infty$. Hence, the SINR after interference cancellation is $1/\tilde{\sigma}_{\sf i}^2$, where $\tilde{\sigma}_{\sf i}^2$ is given in Proposition 3.
\end{proof}

\subsection{SER Analysis}
The unnormalized $\tilde{x}_i$ in \eqref{eq: x_detection} is $\breve{\mathbf{h}}_i^H(\mathbf{h}_ix_i+\tilde{\mathbf{w}}_i)$, where $\tilde{\mathbf{w}}_i$ is
the sum of the additive Gaussian noise and residual interference with the corresponding covariance matrix of $(\sigma^2+\sigma_{\sf i}^2)\mathbf{I}_{N_{\sf r}}$.
Conditioned on $\mathbf{h}_h$ and $\mathbf{h}_t$, the equivalent SNR for symbol detection of $x_i$ can be expressed as
\begin{equation} \label{eq: equivalent noise}
\rho_{i}^{\sf e}{=}\frac{\alpha^{2i}\left|  \left\|\mathbf{h}_h\right\|^2\frac{\alpha^i}{1-\alpha^{2i}} + \mathbf{h}_h^H\mathbf{h}_t\frac{\alpha^j}{1-\alpha^{2j}} \right|  ^2}
{(\sigma^2{+}\sigma_{\sf i}^2{+}1{-}\alpha^{2i})\left|  \mathbf{h}_h^H\mathbf{1}_{N_{\sf r}}\frac{\alpha^i}{1-\alpha^{2i}} + \mathbf{h}_t^H\mathbf{1}_{N_{\sf r}}\frac{\alpha^j}{1-\alpha^{2j}} \right|  ^2},
\end{equation}
where $j=N_{\sf d}+1-i$ and $\sigma_{\sf i}^2$ can be computed from (\ref{eq: covar_c}) or approximated by $\tilde{\sigma}_{\sf i}^2$ in Proposition 3 for large $\rho$.
Thus, the SER at symbol position $i$ can be calculated as
\begin{equation} \label{eq: Pe_i}
P_i^{\sf e}=\int_{}^{}{p(\mathbf{h}_h,\mathbf{h}_t)f_{\sf e}(\rho_{i}^{\sf e})d\mathbf{h}_hd\mathbf{h}_t},
\end{equation}
where $f_{\sf e}(\rho)$ is the error rate corresponding to
instantaneous $\rho$.
For the QPSK modulation,
\begin{equation} \nonumber
f_{\sf e}(\rho){=}\text{erfc}\left(\!\! \sqrt{\rho/2} \right)\!{-}\frac{1}{4} \text{erfc}^2\left( \sqrt{\rho/2} \right),
\end{equation}
and $\text{erfc}(x){=}\frac{2}{\sqrt{\pi}}\int_{x}^{\infty}{e^{-x^2}dx}$
is the complementary error function.
The closed-form expression for $P_i^{\sf e}$ in (\ref{eq: Pe_i}) is difficult to derive.
However, $P_i^{\sf e}$ can be computed accurately by using numerical integration
or by Monte Carlo simulation.
Finally, the overall average SER can be expressed as

\begin{equation} \label{eq: SER}
P^{\sf e}=\frac{1}{N_{\sf d}}\sum_{i=1}^{N_{\sf d}}{P_i^{\sf e}}.
\end{equation}

\subsection{Throughput Analysis}
The throughput is defined as the average number of successfully transmitted data symbol per symbol period, which is averaged over the frame interval. Note that there are $N_{\sf d}$ transmitted data symbols between two consecutive pilot symbols and the frame consists of $N_{\sf p}$ pilot symbols as shown in Fig. \ref{fig: pilot}. Considering the average SER $P^{\sf e}$ in \eqref{eq: SER}, the throughput can be calculated as
\begin{equation} \label{eq: throughput}
{\sf TP} = (1-P^{\sf e})\frac{N_{\sf d}(N_{\sf p}-1)}{(N_{\sf d}+1)(N_{\sf p}-1)+1},
\end{equation}
where, the numerator of the second term of \eqref{eq: throughput} is the number of data symbols transmitted,
and the denominator is the frame length.

The pilot density is defined as $1/(N_{\sf d}+1)$.
It can be verified that when we increase the pilot density (i.e., $N_{\sf d}$ is decreased), $P_{\sf e}$ decreases; thus the first term in (\ref{eq: throughput}) increases. However, the increasing pilot density leads to higher pilot overhead which reduces the second 
term in (\ref{eq: throughput}) and vice versa.
Therefore, there is a trade-off between transmission reliability and throughput,
which suggests that there exists an optimal value of the pilot density that achieves the maximum throughput.

Because the SER in \eqref{eq: Pe_i} and the average SER in \eqref{eq: SER} cannot be expressed in closed form, 
the optimal pilot density for given $\alpha$ and $\rho$
can be found effectively by using the bisection search method.

\subsection{Complexity Analysis}  

	For uncorrelated desired channels, the complexity of our proposed
	interference cancellation, channel estimation and symbol detection
	is linear in the number of antennas,
	since all involved matrix inversions simply become divisions.
	In the first phase, the complexity of EIC estimation is $\mathcal{O}(N_{\sf r}N_{\sf p}^2)$
	and the complexity of channel estimations at pilot positions is $\mathcal{O}(N_{\sf r}N_{\sf p})$.
	In the second phase, while the exhaustive-search based symbol detection approach
	has the complexity growing exponentially with the number of data symbols and the constellation size,
	our proposed I-MAP detection does not depend on the constellation size and has linear complexity in the number of data symbols.
	Particularly, the complexity of the I-MAP detection is $\mathcal{O}(N_{\sf r}N_{\sf d}N_{\sf p})$ which is also linear in the frame length.
	Therefore, the overall complexity of the proposed two-phase design with I-MAP is $\mathcal{O}\left(N_{\sf r}N_{\sf p}(N_{\sf p}+N_{\sf d})\right)$.
    The complexity of the iterative method presented in Section III.C is $\mathcal{O}\left(IN_{\sf r}N_{\sf p}^2N_{\sf d}^2\right)$\footnote{In order to obtain this result, we note that the number of considered pilot symbols in the iterative method is equal to the frame length.}, where $I$ is the average number of iterations to achieve convergence.

\section{Numerical Results}

\subsection{Simulation Settings}

We consider the simulation setting in which the desired receiver has $N_{\sf r}=2$ antennas,
the coefficient $\alpha$ is chosen in the set 
$\left\lbrace 0.95, 0.97, 0.99, 0.995, 0.999\right\rbrace $\footnote{In Clarke's mode, $\alpha=J_0(2\pi f_D T^{\sf d})$,
		where $f_D$ is the maximum Doppler spread  \cite{tse2005fundamentals} (recall that $T^{\sf d}$ is the symbol period 
		of the desired signal).
		Specifically, $\alpha=0.999$ corresponds to 150 Hz of Doppler
		spread with symbol rate of 15 Kbps.
		If the desired signal is carried at 900MHz, the corresponding velocity of the desired Rx is 50m/s.}.
The bandwidth of the interfering signal is two times of the bandwidth of the desired signal, which are $30kHz$ and $15kHz$, respectively.
The frequency spacing $\Delta_f$ between interfering and desired signals will be normalized
as $\Delta_f T^{\sf d}$ where $T^{\sf d}$ denotes the symbol time of the desired signal.
We assume that the QPSK modulation is employed;
both interfering and interfered signals use the root-raised-cosine pulse
shaping function.
Moreover, the pulse shaping functions $p^{\sf d}(t)$ and $p^{\sf i}(t)$ 
are assumed to have
the roll-off factor equal to $0.25$.

The interference power is set as strong as the power of the desired signal and the frequency 
spacing $\Delta_f =1/T^{\sf d}$ unless stated otherwise.
The number of pilot symbols is set equal to 51. Moreover,
the pilot density is chosen in the set $\left\lbrace 25\%, 10\%\right\rbrace$ 
corresponding to $\left\lbrace3,9\right\rbrace$  data symbols between two pilot symbols, respectively.
Furthermore, for throughput simulation results, we show the throughputs obtained for various pilot densities ranging from $50\%$ to $6.25\%$.
The results presented in this section are obtained by 
averaging over $10^4$ random realizations.

\subsection{Performance of the Proposed Channel Estimation Technique}
For the interference-free scenario, we investigate the effect of different parameters to the channel estimation errors.
We note that the performance of the channel estimation technique presented in this section depends mainly on $N_{\sf d}$ and $\alpha$.
	Specifically, the performance depends on $\alpha_{\sf p}$ which is the correlation coefficient of
	channel gains at two consecutive pilot positions (see Appendix A and Theorem 1).
	Different values of $N_{\sf d}$ (different pilot densities) have the corresponding values of $\alpha_{\sf p}$.
We will show the numerical channel estimation mean squared error (CMSE)
which is calculated as
\begin{equation}
\begin{split}
{\sf CMSE}
&=\frac{1}{N_{\sf p}N_{\sf r}} \sum_{n=1}^{N_{\sf p}}{
\textbf{tr}\left(\E \left[ \left( \mathbf{h}_n-\tilde{\mathbf{h}}_n\right)\left( \mathbf{h}_n-\tilde{\mathbf{h}}_n\right)^H \right]
\right)}.\\
\end{split}	
\end{equation}

\begin{figure}[ht]
	\centering
	\includegraphics[scale=0.275]{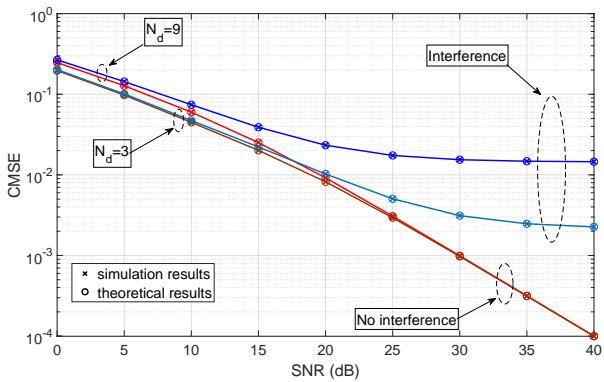}
	\caption{Channel estimation mean squared error, $\alpha=0.99$}
	\label{fig: noniterative CMSE vs SNR}
\end{figure}

In Fig. \ref{fig: noniterative CMSE vs SNR}, we show the channel estimation error due to our proposed
design for different values of $N_{\sf d}$ (equivalently, different values of pilot density),
when there is no interference (IF) and when there is interference (IP).
When $N_{\sf d}$ increases, the channel estimation mean squared error also increases as expected.
For the interference-free scenario, the corresponding error curves converge to each other and decrease almost linearly
as the SNR increases (both curves are plotted in the log scale). This
means that the impact of the fast fading is diminished in the high SNR regime.
When the interference is present, there is a performance floor for channel estimation error.
The results in Fig. \ref{fig: noniterative CMSE vs SNR} also validate the theoretical results stated in Propositions 1, 3,
and 4 about the channel estimation errors in the scenarios without and with interference.

\begin{figure}[ht]  
	\centering
	\includegraphics[scale=0.37]{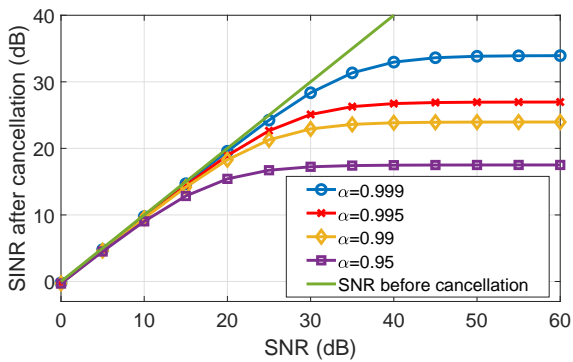} 
	\caption{SINR after cancellation for different values of channel correlation coefficient, $N_{\sf d}=3$}  
	\label{fig: SINR after cancellation}
\end{figure} 

In Fig. \ref{fig: SINR after cancellation}, we show the achieved SINR after interference cancellation versus the SNR for different values of channel correlation coefficient $\alpha$.
Two noticeable observations can be drawn from this figure.
First, it can be seen that the achieved SINR increases with increasing SNR before becoming saturated. In the low SNR regime, however, the residual interference  has almost no impact on the achieved SINR after interference cancellation, i.e., the SINR curves after interference cancellation are very close to the line showing the SNR before interference cancellation.
	Second, the achieved SINR after cancellation increases with the increasing values of channel correlation coefficient $\alpha$.
	This is because the higher the value of $\alpha$ is,
	the lower the variance of the channel evolutionary noise
	and the less severe the impact of the fast fading are. 	
Since the fast fading noise is less disruptive, interference cancellation performance is alleviated (as it is known that the fast fading noise causes the performance floor for the interference cancellation),
	which in turn reduces the residual interference power and makes the achieved SINR higher.

\begin{figure}[ht]  
	\centering
	\includegraphics[scale=0.37]{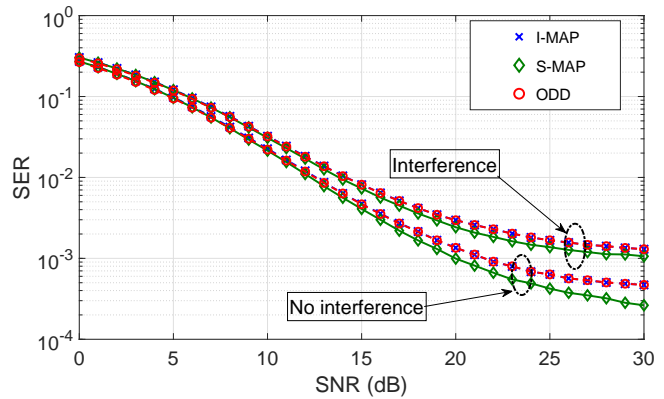} 
	\caption{SER achieved by different detection methods, $N_{\sf d}=3$}   
	\label{fig: SER}
\end{figure}

\subsection{Performance of the Proposed Symbol Detection Methods}
We now compare the SER performance of series symbol MAP detection (S-MAP), individual symbol MAP detection (I-MAP)
and optimum diversity detection (ODD) \cite{sun2014maximizing, mahamadu2018fundamental} methods.
The ODD method is the optimum individual symbol detection with imperfect CSI.
	Basically, in the ODD method, the channel gains at data positions are interpolated from the MMSE-estimated channel gains at pilot positions.
	Then, the zero-forcing based symbol detection is employed (please refer to Sections III and IV in \cite{mahamadu2018fundamental} for 
	more details).

Fig. \ref{fig: SER} illustrates the SER achieved by these detection methods for the interference-free and interference scenarios,
which are denoted as IF and IP in this section, respectively.
It can be seen that the SER of the proposed I-MAP is almost identical to that achieved by the ODD method.
Moreover, the S-MAP detector outperforms both I-MAP and ODD and
the performance gap is larger in the interference-free scenario.
Note that, in the IP scenario, the residual interference still presents, which causes the error floors in these SER curves.

\begin{figure}[ht]  
	\centering
	\includegraphics[scale=0.375]{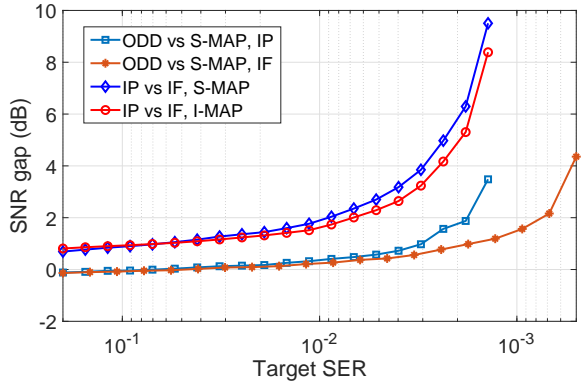} 
	\caption{SNR gap for specific target SER, $N_{\sf d}=3$}
	\label{fig: SNRgap}
\end{figure} 

For performance comparison between our methods and the existing method,
we show in Fig. \ref{fig: SNRgap} the SNR gap to achieve the same SER between different symbol detection methods (S-MAP, I-MAP) and scenarios (IF, IP).
Particularly,
	a value of 3dB SNR gap at $5 \times 10^{-3}$ target SER of the curve \textit{A vs B} means that method A needs 3dB higher in SNR to achieve the same target SER achieved by method B.
For the same scenario (IF or IP), the SNR gap between the proposed S-MAP and ODD becomes larger as the required SER decreases.
Note again that there is a performance floor in the IP scenario;
nevertheless, our proposed detection method achieves more than 3dB SNR gain compared to the existing ODD method for the 
same detection performance in the low target SER regime (see the curve with square markers).
Moreover, to achieve the same SER performance under the high reliability condition (i.e., low SER),
the SNR required in the interference scenario is much higher than that required 
in the interference free scenario (illustrated by IP vs IF curves).

\begin{figure}[ht]  
	\centering
	\includegraphics[scale=0.375]{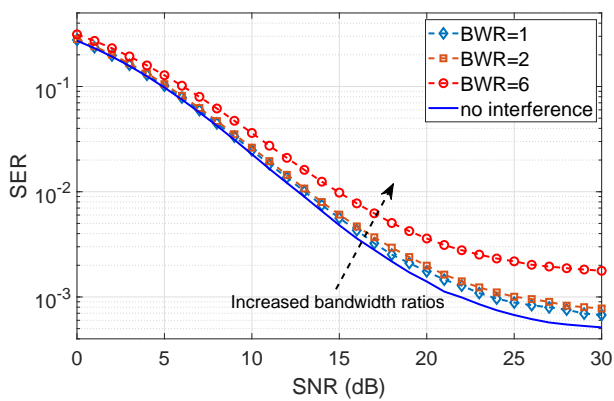} 
	\caption{SER versus SNR for different values of BWR, $N_{\sf d}=3$}
	\label{fig: BWratio}
\end{figure} 

Fig. \ref{fig: BWratio} illustrates the SER in the interference-free and interference scenarios 
for different bandwidth ratios, which is denoted as BWR.
As can be seen from this figure, higher bandwidth ratios between interfering and desired signals
lead to higher SER. This is because higher BWR creates more severe interference for the desired signal
and it is not possible to completely remove the interference due to the fast fading.

\begin{figure}[ht]  
	\centering
	\includegraphics[scale=0.375]{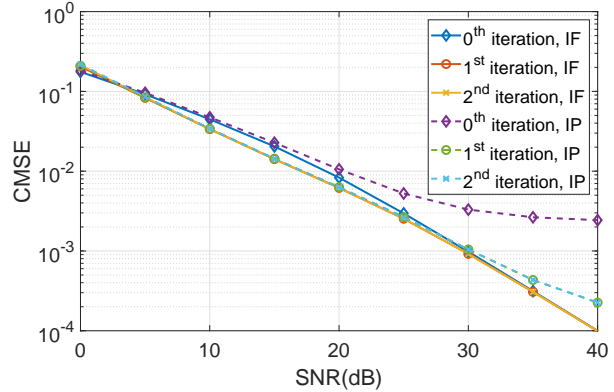}
	\caption{Performance of channel estimation for iterative algorithm} 
	\label{fig: iterative CMSE vs SNR}
\end{figure} 

\subsection{Performance of the Iterative Algorithm}
We now study the performance of the iterative algorithm for channel estimation, interference cancellation, and symbol detection.
First, we present the performance of channel estimation over iterations
in Fig. \ref{fig: iterative CMSE vs SNR} where the CMSE of estimated channel gains
is shown for both IF and IP scenarios.
As can be seen from this figure, the iterative algorithm converges\footnote{In the simulation,  the convergence is 
actually achieved when there is no change in the detected data symbols. For a better illustration, we show 
the `convergence' of the CMSE instead. This is because there is no change in the estimated channel if there is no 
change in the detected data symbols over iterations.}
after only a few iterations.
The most noticeable observation is that
the converged channel estimation performance  in the presence of interference (IP) is almost identical to that 
of the interference free scenario (IF) in the low SNR regime (less then 30dB),
which implies that the proposed iterative method cancels very well the interference in this SNR region.
When the SNR is higher than 30dB, the performance in the IP case is still limited by the fast fading noise.
However, the performance floor of the iterative channel estimation approach is much lower than that of the non-iterative counterpart (the $0^{th}$-iteration\footnote{Note that iterations are only counted when the algorithm enters the while loop. In other words, results obtained from the first and second steps in Algorithm 3 are considered at the $0^{th}$ iteration. In Algorithm 3, we choose I-MAP due to its low complexity, 
but S-MAP can also be used.} versus the $2^{nd}$-iteration curves in the IP scenario).

We now study how the SER improves over iterations.
In Fig. \ref{fig: ser_iterative}, the left and right figures show the SERs for the IF and IP cases, respectively.
It can be seen from the figure that the SER improvement is higher when the interference is present,
which suggests that the iterative algorithm estimates and cancels the interference effectively.

We show the SERs achieved by the non-iterative and iterative 
algorithms\footnote{The SER of the non-iterative algorithm is the SER computed at the $0^{th}$ iteration and the SER of the iterative algorithm is the SER achieved at convergence.}.
From Fig. \ref{fig: ser_ipif}, we can see that the iterative algorithm improves the SER in both IF and IP scenarios.
Furthermore, the improvement is higher for larger values of SNR.
This is because that the high SNR regime allows more reliable data detection, which boosts the performance of 
interference cancellation and channel estimation.

\begin{figure}[ht]  
	\centering
	\includegraphics[scale=0.355]{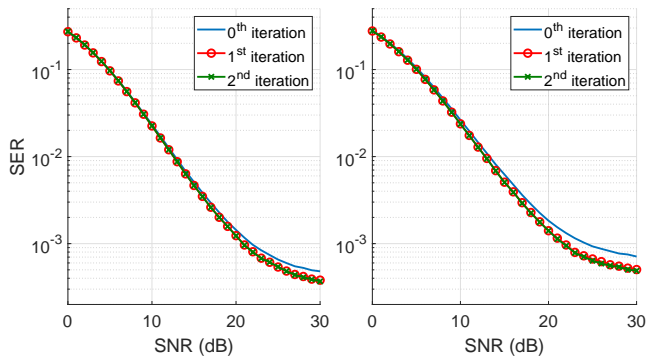} 
	\caption{SER over iterations} 
	\label{fig: ser_iterative}
\end{figure} 

\begin{figure}[ht]  
	\centering
	\includegraphics[scale=0.35]{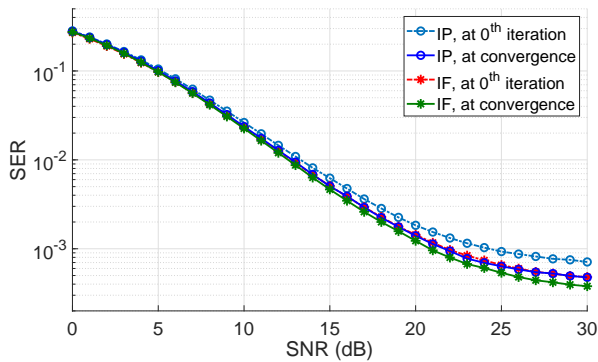} 
	\caption{SER achieved by iterative and non-iterative algorithms} 
	\label{fig: ser_ipif}
\end{figure} 

\subsection{Throughput Achieved by Proposed Framework}
In Fig. \ref{fig: throughput}, we show the variations of the throughput with the pilot density for different values of SNR $\rho$
and channel correlation coefficient $\alpha$. 
As can be seen from this figure, for given $\alpha$ and $\rho$, there exists an optimal pilot density that 
achieves the maximum throughput. Moreover, the maximum throughput increases as the SNR $\rho$ increases.
It can also be observed that larger $\alpha$ leads to higher maximum throughput and lower optimal pilot density.
This is because when the channel varies more slowly, the performance of interference cancellation and channel estimation is
improved, which results in more reliable transmission and higher throughput.
The results in this figure demonstrate the tradeoff between the throughput and communication reliability in the fast fading environment.

\begin{figure}[ht]  
	\centering
	\includegraphics[scale=0.375]{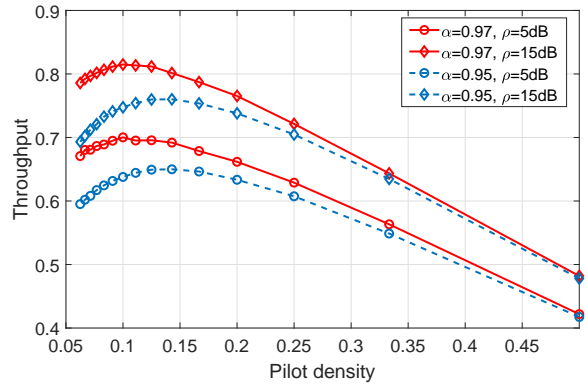} 
	\caption{Throughput variations with the pilot density}  
	\label{fig: throughput}
\end{figure} 

\section{Conclusion}  
We have proposed two frameworks for channel estimation, interference cancellation, 
and symbol detection for communication signals with different bandwidth in the fast fading environment.
Specifically, in the two-phase non-iterative framework, 
we have derived the channel estimators and studied both series and individual symbol detection methods.
The iterative framework performs interference cancellation, channel estimation and data detection
based on the detected data symbol from the previous iteration, which can improve the system performance
compared to the non-iterative counterpart.
Numerical studies have confirmed the existence of the performance floor for SER in the considered
interference scenario and illustrated the optimal pilot density to achieve the maximum throughput.
Moreover, we have shown that the series symbol detection method outperforms the existing ODD method in terms of SER while
the individual symbol detection method achieves the very close performance to the 
ODD method but with lower complexity.

\appendices
\section{Proof of Theorem 1}
To compute $p(\mathbf{h}_n,\mathbf{Y})$,
we need to find $p(\mathbf{Y}|\mathbf{h}_n)$, since
\begin{equation} \label{eq: prob h_n, Y}
p(\mathbf{h}_n, \mathbf{Y})=p(\mathbf{Y}|\mathbf{h}_n)p(\mathbf{h}_n),
\end{equation}
and $p(\mathbf{h}_n)$ is known to be $\mathcal{CN}(\mathbf{h}_n,\mathbf{0},\mathbf{I}_{N_{\sf r}})$,
where $\mathcal{CN}(\textbf{x},\boldsymbol{\mu},\boldsymbol{\Sigma})$ is the complex Gaussian density of random vector $\textbf{x}$ having mean $\boldsymbol{\mu}$ and covariance matrix $\boldsymbol{\Sigma}$ \cite{zhu2009message}.
The likelihood of $\mathbf{Y}$, given $\mathbf{h}_n$ can be factorized, thanks to the channel Markovian property, as  
\begin{equation}\label{eq: AL1_LLH}	
p(\mathbf{Y}|\mathbf{h}_n){=}p(\mathbf{y}_n|\mathbf{h}_n) \!
\prod_{i=1}^{n-1}{p(\mathbf{y}_{i}|\mathbf{y}_{i{+}1},\mathbf{h}_n)}\!\!\!
\prod_{i=n+1}^{N_{\sf p}}{\!\! p(\mathbf{y}_{i}|\mathbf{y}_{i{-}1},\mathbf{h}_n)}.
\end{equation}

Given $\mathbf{h}_n$, any two consecutive observations are correlated due to the cumulative channel evolutionary noises.
Since we consider only received signals at pilot positions, the equivalent correlation coefficient of channel gains at two consecutive pilot positions is
$\alpha_{\sf p}=\alpha^{N_{\sf d}+1}$.
To further derive $p(\mathbf{Y}|\mathbf{h}_n)$, we need to find the probabilities $p(\mathbf{y}_{i}|\mathbf{y}_{i-1},\mathbf{h}_n)$ for $i>n$
and $p(\mathbf{y}_{i}|\mathbf{y}_{i+1},\mathbf{h}_n)$ for $i<n$.

We now show the derivation of $p(\mathbf{y}_{i}|\mathbf{y}_{i-1},\mathbf{h}_n)$ for $i>n$.
From (\ref{eq: channel_model}), the channel coefficient $\mathbf{h}_i, i>n$ can be expressed with respect to
$\mathbf{h}_n$ as 
\begin{equation} \label{eq: hi to hn, i>n}
\mathbf{h}_i=\alpha_{\sf p}^{i-n}\left( \mathbf{h}_n+\eta_{\sf p}\sum_{j=1}^{i-n}{\alpha_{\sf p}^{-j}\boldsymbol{\Delta}_{n+j}}\right) , \\
\end{equation}
where $\eta_{\sf p}= \left( 1{-}\alpha_{\sf p}^2\right) ^{1/2}$.
Substituting $\mathbf{h}_n$ in (\ref{eq: hi to hn, i>n}) into (\ref{eq: Rx final}),
it can be seen that $\mathbf{y}_i$ and $\mathbf{y}_{i-1}$ share the common evolutionary noise terms 
$\boldsymbol{\Delta}_{n+j},  j=1,...,i{-}n{-}1$.
Then, we can obtain the parameters of the distribution $p(\mathbf{y}_{i},\mathbf{y}_{i-1}|\mathbf{h}_n)=$
$\mathcal{CN}\left(\! \left[ \!\! \begin{array}{lll}
\mathbf{y}_i\\
\mathbf{y}_{i-1}
\end{array}\!\!\!\right] 
\!\!,\!\!
\left[ \!\! \begin{array}{lll}
\boldsymbol{\mu}_{\mathbf{y}_i|\mathbf{h}_n}\\
\boldsymbol{\mu}_{\mathbf{y}_{i-1}|\mathbf{h}_n}
\end{array} \!\! \right]
\!\!,\!\!
\left[ \!\! \begin{array}{lll}
\boldsymbol{\Sigma}_{\mathbf{y}_i|\mathbf{h}_n} \!\!&\!\!\! \boldsymbol{\Sigma}_{\mathbf{y}_i,\mathbf{y}_{i-1}|\mathbf{h}_n}\\
\boldsymbol{\Sigma}_{\mathbf{y}_i,\mathbf{y}_{i-1}|\mathbf{h}_n}^H \!\! & \!\!\! \boldsymbol{\Sigma}_{\mathbf{y}_{i-1}|\mathbf{h}_n}
\end{array} \!\! \right]
\right) $ as follows:
\begin{equation}
\begin{split}
\boldsymbol{\mu}_{\mathbf{y}_k|\mathbf{h}_n}
&{=}\mathbf{B}_k\mathbf{c}+\alpha_{\sf p}^{k-n}\mathbf{h}_nx_k, \; k = i, i-1,\\ 
\boldsymbol{\Sigma}_{\mathbf{y}_k|\mathbf{h}_n}
&{=}\left( \sigma^2{+}\alpha_{\sf p}^{2(k-n)}\eta_{\sf p}^2\sum_{j=1}^{k-n}{\alpha_{\sf p}^{-2j}} \right)\mathbf{I}_{N_{\sf r}}\\
&{=}\left[ 1{+}\rho\left( 1-\alpha_{\sf p}^{2(k-n)}\right)  \right] \sigma^2\mathbf{I}_{N_{\sf r}}, \; k = i, i{-}1,\\
\boldsymbol{\Sigma}_{\mathbf{y}_i,\mathbf{y}_{i{-}1}|\mathbf{h}_n}
&{=} \E \left[ \left( \mathbf{y}_i{-}\boldsymbol{\mu}_{\mathbf{y}_i|\mathbf{h}_n} \right) \left( \mathbf{y}_{i-1}{-}\boldsymbol{\mu}_{\mathbf{y}_{i-1}|\mathbf{h}_n} \right)^H | \mathbf{h}_n\right]\\
&{=}x_ix_{i-1}^{*}\alpha_{\sf p}(1-\alpha_{\sf p}^{2(i-n-1)})\mathbf{I}_{N_{\sf r}}.
\end{split}
\end{equation}

Next, we apply the conditional probability formula for the multivariate Complex Circular Symmetric Gaussian vector \cite{gallager2013stochastic} (section 3.7.7, page 153) and obtain $ p(\mathbf{y}_{i}|\mathbf{y}_{i-1},\mathbf{h}_n)=
\mathcal{CN}(\mathbf{y}_{i},\boldsymbol{\mu}_{i,n},\boldsymbol{\Sigma}_{i,n})$  for  $i>n$,
where
\begin{equation}\label{eq: cond_para_after}
\begin{split}
\boldsymbol{\mu}_{i,n}
&{=}\boldsymbol{\mu}_{\mathbf{y}_i|\mathbf{h}_n}+\beta_{i,n}
\left( \mathbf{y}_{i-1}-\boldsymbol{\mu}_{\mathbf{y}_{i-1}|\mathbf{h}_n} \right), \\
\boldsymbol{\Sigma}_{i,n}
&{=}\sigma_{i,n}^{2} \mathbf{I}_{N_{\sf r}}, \;
\beta_{i,n}=\frac{x_{i}x_{i-1}^*\rho\alpha_{\sf p} \left( 1-\alpha^{2(i - n - 1)}_{\sf p}\right) }{1+\rho \left( 1-\alpha_{\sf p}^{2(i-n-1)}\right) },\\
\sigma_{i,n}^2
&{=}\sigma^2\!\!\left[\!1 {+}\rho \! \left(\! 1{-}\alpha_{\sf p}^{2(i-n)}\! \right) \! {-}
\frac{\rho^2\alpha^2_{\sf p} \left( 1{-}\alpha_{\sf p}^{2(i-n-1)}\right) ^2}{1{+}\rho \left( 1{-}\alpha_{\sf p}^{2(i-n-1)}\right) }\right].
\end{split} 
\end{equation}

For $i<n$, $p(\mathbf{y}_{i}|\mathbf{y}_{i+1},\mathbf{h}_n)=
\mathcal{CN}(\mathbf{y}_{i},\boldsymbol{\mu}_{i,n},\boldsymbol{\Sigma}_{i,n})$,
where the parameters can be expressed similarly:
\begin{equation} \nonumber
\boldsymbol{\mu}_{i,n}
{=}\alpha^{n-i}_{\sf p}\mathbf{h}_nx_{i}{+}\mathbf{B}_{i}\mathbf{c}
{+} 	\beta_{i,n}(\mathbf{y}_{i+1}{-}\alpha_{\sf p}^{n-i-1}\mathbf{h}_nx_{i+1}{-}\mathbf{B}_{i+1}\mathbf{c}),
\end{equation}
\begin{equation}\label{eq: cond_para_before}
\begin{split}
\boldsymbol{\Sigma}_{i,n}
&{=}\sigma_{i,n}^2 \mathbf{I}_{N_{\sf r}}, \; \;
\beta_{i,n} =\frac{x_{i}x_{i+1}^*\rho\alpha_{\sf p} \left(1- \alpha^{2(n-i-1)}\right) }{1+\rho \left(1- \alpha_{\sf p}^{2(n-i-1)}\right) },\\
\sigma_{i,n}^2
&{=}\sigma^2\!\!\left[\!1 {+}\rho \! \left(\! 1{-}\alpha_{\sf p}^{2(n-i)}\!\right)\!
{-}
\frac{\rho^2\alpha^2_{\sf p} \left( 1{-}\alpha_{\sf p}^{2(n-i-1)}\right) ^2}{1{+}\rho \left( 1{-}\alpha_{\sf p}^{2(n-i-1)}\right) }\right] .
\end{split}
\end{equation}

For $j=n,
\boldsymbol{\mu}_{n,n}{=} \mathbf{h}_nx_n + \mathbf{B}_n\mathbf{c}, \;
\boldsymbol{\Sigma}_{n,n}{=}\sigma^2\mathbf{I}_{N_{\sf r}}$.
Substituting the parameters in (\ref{eq: cond_para_after}) and (\ref{eq: cond_para_before}) into (\ref{eq: prob h_n, Y}) using  (\ref{eq: AL1_LLH}), taking the logarithm, we obtain the  log-likelihood function in Theorem 1. This completes the proof.

\section{Proof for the Positive-definiteness of $\mathbf{D}_n$}
For an arbitrary non-zero vector $\mathbf{z}=[z_1,...,z_L]^T$, we have
\begin{equation} \label{eq: prood_D}
\begin{split}
\mathbf{z}^H\mathbf{D}_n\mathbf{z}
&{=}\text{tr}\left[ 
\left( \sum_{i=1}^{N_{\sf p}}{\frac{\left(\mathbf{B}_{i,n}\mathbf{z} \right) ^H\left(\mathbf{B}_{i,n}\mathbf{z} \right) }{\sigma_{i,n}^2}}\right)\right] \\
-\text{tr}&\left[ \left(\sum_{i=1}^{N_{\sf p}}{\frac{x^*_{i,n}\mathbf{B}_{i,n}}{\sigma_{i,n}^2}} \mathbf{z} \right)^{\!H} \!\! \mathbf{A}_n^{-1}  \left(\sum_{i=1}^{N_{\sf p}}{\frac{x^*_{i,n}\mathbf{B}_{i,n}}{\sigma_{i,n}^2}} \mathbf{z} \right)\right]\\
&{=}\text{tr}\left[ 
\left( \sum_{i=1}^{N_{\sf p}}{\frac{\left(\mathbf{B}_{i,n}\mathbf{z} \right) \left(\mathbf{B}_{i,n}\mathbf{z} \right)^H }{\sigma_{i,n}^2}}\right)\right] \\
-\text{tr}&\left[ \left(\sum_{i=1}^{N_{\sf p}}{\frac{x^*_{i,n}\mathbf{B}_{i,n}}{\sigma_{i,n}^2}} \mathbf{z} \right)\mathbf{A}_n^{-1}  \left(\sum_{i=1}^{N_{\sf p}}{\frac{x^*_{i,n}\mathbf{B}_{i,n}}{\sigma_{i,n}^2}} \mathbf{z} \right)^{\!H}\right]\\
& {>} \text{tr}\left[ 
\left( \sum_{i=1}^{N_{\sf p}}{\frac{\left(\mathbf{B}_{i,n}\mathbf{z} \right) \left(\mathbf{B}_{i,n}\mathbf{z} \right)^H }{\sigma_{i,n}^2}}\right)\right. \\
-&
\left. \! \!
\left( \! \sum_{i=1}^{N_{\sf p}}{\frac{x^*_{i,n}\mathbf{B}_{i,n}}{\sigma_{i,n}^2}} \mathbf{z}\!
\right)
\! \!
\left( \mathbf{A}_n {-}\mathbf{I}_{N_{\sf r}}\right) ^{\! {-}1}
\!\!
\left(\! \sum_{i=1}^{N_{\sf p}}{\frac{x^*_{i,n}\mathbf{B}_{i,n}}{\sigma_{i,n}^2}} \mathbf{z} \!\!
\right)^{\!\! \!\! H}
\right]\! ,
\end{split}
\end{equation}
where $\text{tr}(\mathbf{X})$ is the sum of diagonal elements of $\mathbf{X}$.
In the last two lines of (\ref{eq: prood_D}),
the $j$th diagonal element of the first term is $\sum_{i=1}^{N_{\sf p}}{(\mathbf{b}_{i,n}^{(j)}\mathbf{z})(\mathbf{b}_{i,n}^{(j)}\mathbf{z})^H/\sigma_{i,n}^2}$,
where $\mathbf{b}_{i,n}^{(j)}$ is the $j$th row of $\mathbf{B}_{i,n}$,
and the $j$th diagonal element of the second term is $\left( \sum_{i=1}^{N_{\sf p}}{\frac{|x_{i,n}|^2}{\sigma_{i,n}^2}}\right) ^{-1}\left(\sum_{i=1}^{N_{\sf p}}{\frac{x_{i,n}^*\mathbf{b}_{i,n}^{(j)}\mathbf{z}}{\sigma_{i,n}^2}} \right)\left(\sum_{i=1}^{N_{\sf p}}{\frac{x_{i,n}^*\mathbf{b}_{i,n}^{(j)}\mathbf{z}}{\sigma_{i,n}^2}} \right)^H $, where $\mathbf{A}_n$ from (\ref{eq: step1_para}) is substituted into this term.

We now define the two vectors $\mathbf{u}$ and $\mathbf{v}$ whose $i$th elements are 
$\text{u}_i=x_{i,n}^*/\sigma_{i,n}, \text{v}_i=\mathbf{b}_{i,n}^{(j)}\mathbf{z}/\sigma_{i,n}$, respectively.
By applying the \textit{Cauchy}-\textit{Schwarz} inequality
\begin{equation} \label{eq: Cauchy}
|\mathbf{u}|^2|\mathbf{v}|^2 \geq |\mathbf{u}.\mathbf{v}|^2,
\end{equation}
it can be verified that each diagonal element of the matrix
in the last two lines
of (\ref{eq: prood_D}) is positive, which means its trace is also positive. Thus, we have completed the proof.

\section{Proof of Theorem 2}
We can reformulate $p\left( x_{1:N_{\sf d}}|\mathbf{h}_h,\mathbf{h}_t,\mathbf{y}_{1:N_{\sf d}} \right)$ as follows:
\begin{equation} \label{eq: post_pro}
\begin{split}
&p\left( x_{1:N_{\sf d}}|\mathbf{h}_h,\mathbf{h}_t,\mathbf{y}_{1:N_{\sf d}} \right)\\
&{\propto}
\int_{}^{}{p(x_{1:N_{\sf d}},\mathbf{y}_{1:N_{\sf d}},\mathbf{h}_h,\mathbf{h}_{1:N_{\sf d}},\mathbf{h}_t)
	d\mathbf{h}_{1:N_{\sf d}}}\\
& {\stackrel{(a)}{\propto}} 
\int_{}^{}{p(\mathbf{y}_{1:N_{\sf d}}|\mathbf{h}_{1:N_{\sf d}},x_{1:N_{\sf d}})p(\mathbf{h}_h,\mathbf{h}_{1:N_{\sf d}},\mathbf{h}_t)d\mathbf{h}_{1:N_{\sf d}}}\\
&{\stackrel{(b)}{\propto}}  \int_{}^{}{p(\mathbf{h}_h,\mathbf{h}_{1:N_{\sf d}},\mathbf{h}_t)\prod_{i=1}^{N_{\sf d}}{p(\mathbf{y}_i|x_i,\mathbf{h}_i)}
	d\mathbf{h}_{1:N_{\sf d}}}\\
&{\stackrel{(c)}{\propto}} \int_{}^{}{p(\mathbf{h}_1|\mathbf{h}_h)p(\mathbf{h}_t|\mathbf{h}_{N_{\sf d}})
	\prod_{i=2}^{N_{\sf d}}{p(\mathbf{h}_i|\mathbf{h}_{i-1})}}\\
& \quad \quad \quad \quad \quad \quad \quad \quad \quad \quad
	\prod_{i=1}^{N_{\sf d}}{p(\mathbf{y}_i|x_i,\mathbf{h}_i)}
	d\mathbf{h}_{1:N_{\sf d}}\\
&{\stackrel{(d)}{\propto}} \; e^{\mathcal{F}}\int_{}{}
{exp\left\lbrace - \sum _{i=1}^{N_{\sf d}} {(\mathbf{h}_i-\mathbf{a}_i)^H\mathbf{S}_i^{-1}(\mathbf{h}_i-\mathbf{a}_i)}
	\right\rbrace d\mathbf{h}_{1:N_{\sf d}}}\\
&{\stackrel{(e)}{\propto}} \; e^{\mathcal{F}} \prod_{i=1}^{N_{\sf d}}{|\mathbf{S}_i|}, \\
\end{split}
\end{equation}
where
$
\boldsymbol{\Gamma}_{i,j}=\tau_2^{i-j}\prod_{k=j}^{i-1}{\mathbf{S}_k},
\tau_1=\frac{1}{(1-\alpha^2)\sigma^2_{\sf h}}, \tau_2=\alpha \tau_1,
$ and
$\mathcal{F}$ is defined in (\ref{eq: F}).
Assuming all points in the constellation are transmitted with equal
probability, 
the conditional probabilities in (\ref{eq: post_pro}) are transformed by
Baye's rule $(a)$ and the Markovian property of channel $(b,c)$,
where these expressions can be obtained by iteratively synthesizing quadratic terms of $\mathbf{h}_i,i=1,...,N_{\sf d}$ in the exponents $(d,e)$.

\section{Proof of Proposition 1}
The channel estimation error can be written as (see \eqref{eq: step1_para}, \eqref{eq: xyB})
\begin{equation} \label{eq: CEN decomposed}
\begin{split}
\boldsymbol{\nu}_n&= \mathbf{h}_n-\mathbf{A}_n^{-1}\left( \sum_{i=1}^{N_{\sf p}}{\frac{x_{i,n}^*}{\sigma_{i,n}^2}}\left( \mathbf{y}_i-\beta_{i,n}\mathbf{y}_{i+j_{i,n}} \right)  \right)\\
&= \boldsymbol{\nu}_n^{\sf g}+\boldsymbol{\nu}_n^{\sf c},
\end{split}
\end{equation}
where $\boldsymbol{\nu}_n^{\sf g}$ is the error due to the AWGN, and $\boldsymbol{\nu}_n^{\sf c}$ is the error due to the channel evolutionary noise. Specifically,

\begin{equation} \label{eq: CEN decomposed multipliers}
\begin{split}
\boldsymbol{\nu}_n^{\sf g}&=\sum_{i=1}^{N_{\sf p}}{\boldsymbol{\Xi}_{i,n}^{\sf g}\mathbf{w}_i},\\ 
\boldsymbol{\nu}_n^{\sf c}&=\boldsymbol{\Xi}_{0,n}^{\sf c}\mathbf{h}_0+ \sum_{i=1}^{N_{\sf p}}{\boldsymbol{\Xi}_{i,n}^{\sf c}\boldsymbol{\Delta}_i},
\end{split}
\end{equation}
where we decompose $\mathbf{h}_n$ into $\alpha_{\sf p}^n\left(\mathbf{h}_0+\sum_{i=1}^{n}{\alpha_{\sf p}^{-i}\boldsymbol{\Delta}_i}\right) $.
By using this decomposition, it is more convenient to compute the channel estimation error components due to the channel evolutionary noise.
Otherwise, one has to determine the dependence structure of $\mathbf{h}_n$ on the preceding channel noise components $\boldsymbol{\Delta}_i, i<n$, which is not trivial.

When the desired channels are independent, $\boldsymbol{\Xi}_{i,n}^{\sf g}=\xi_{i,n}^{\sf g}\mathbf{I}_{N_{\sf r}}$ and  $\boldsymbol{\Xi}_{i,n}^{\sf c}=\xi_{i,n}^{\sf c}\mathbf{I}_{N_{\sf r}}$.
Substituting $\mathbf{y}_i=\mathbf{h}_ix_i+\mathbf{w}_i$ into (\ref{eq: CEN decomposed}), we have

\begin{equation} \nonumber
\xi_{i,n}^{\sf g}{=}\left\{
\begin{array}{ll}
\!\!\!\!{-}\frac{x_{i,n}^*}{a_n\sigma_{i,n}^2}, & \!\! \!\!i{=}1{,}n{,}n{\pm} 1{,}N_{\sf p},\\  
\!\!\!\!{-}\frac{x_i^*}{a_n}\left( \frac{\omega_{i,n}}{\sigma_{i,n}^2}{+}\frac{|\beta_{n,i-j_{i,n}}|\omega_{n,i-j_{i,n}}}{\sigma_{n,i-j_{i,n}}^2}\right)\!,&\!\! \!\!\text{otherwise}.
\end{array}
\right.
\end{equation}

Hence, the AWGN contributes to the CEE with the total power of $\sigma^2\sum_{i=1}^{N_{\sf p}}{|\xi_{i,n}^{\sf g}|^2}$.
As the SNR goes to infinity, $\text{lim}_{\rho \to \infty}\sum_{i=1}^{N_{\sf p}}{|\xi_{i,n}^{\sf g}|^2}=1$,  and the AWGN 
contributes $\sigma^2$ to the overall CEE.
Besides, $\boldsymbol{\nu}_n^{\sf c}$ is expressed in (\ref{eq: CEN decomposition c multipliers}), and we can write the multipliers  $\xi_{i,n}^{\sf c}$ as follows:
\begin{equation}
\begin{split}
\xi_{0,n}^{\sf c}&=\alpha_{\sf p}^n-\sum_{i=1}^{N_{\sf p}}{\frac{\omega_{i,n}\alpha_{\sf p}^i}{a_n\sigma_{i,n}^2}(1-|\beta_{i,n}|\alpha_{\sf p}^{j_{i,n}})},\\
\xi_{i,n}^{\sf c}&=\frac{\alpha_{\sf p}^{-i}}{a_n}
\left(
\sum_{k=i-1}^{n-1}{\frac{\omega_{k,n}}{\sigma_{k,n}^2}|\beta_{k,n}|\alpha_{\sf p}^{k+1}}
-\sum_{k=i}^{N_{\sf p}}{\frac{\omega_{k,n}}{\sigma_{k,n}^2}\alpha_{\sf p}^{k}}
\right. \\
& \quad \quad  \left. + \alpha_{\sf p}^n +\sum_{k=n+1}^{N_{\sf p}}{\frac{\omega_{k,n}}{\sigma_{k,n}^2}|\beta_{k,n}|\alpha_{\sf p}^{k-1}}
\right), \; \; i\leq  n,\\
\xi_{i,n}^{\sf c}&=\frac{\alpha_{\sf p}^{-i}}{a_n}
\left(
\sum_{k=i+1}^{N_{\sf p}}{\frac{\omega_{k,n}}{\sigma_{k,n}^2}|\beta_{k,n}|\alpha_{\sf p}^{k-1}}
-\sum_{k=i}^{N_{\sf p}}{\frac{\omega_{k,n}}{\sigma_{k,n}^2}\alpha_{\sf p}^{k}}
\right.\\
& \quad \left. -\frac{1}{\sigma^2}\mathds{1}_{n<N}\mathds{1}_{i=N}
\right), \quad i>n.
\end{split}
\end{equation}

\newcounter{tempequationcounter_3}
\begin{figure*}[ht]
	\normalsize
	\setcounter{tempequationcounter_1}{\value{equation}}
	\begin{equation} \label{eq: CEN decomposition c multipliers}
	\begin{split}
	\boldsymbol{\nu}_n^{\sf c}&=\mathbf{h}_n
	-\mathbf{A}_n^{-1} \left\lbrace \sum_{i=1}^{N_{\sf p}}{\frac{\omega_{i,n}\alpha_{\sf p}^i}{\sigma_{i,n}^2}
		\left[ 
		\mathbf{h}_0+
		\sum_{k=1}^{i}{\alpha_{\sf p}^{-k}\boldsymbol{\Delta}_{k}} 
		-|\beta_{i,n}|\alpha_{\sf p}^{j_{i,n}} \left(\mathbf{h}_0+
		\sum_{k=1}^{i+j_{i,n}}{\alpha_{\sf p}^{-k}\boldsymbol{\Delta}_{k}}\right) 
		\right]  }\right\rbrace \\
	&=\mathbf{h}_0\left( \alpha_{\sf p}^n-\sum_{i=1}^{N_{\sf p}}{\frac{\omega_{i,n}\alpha_{\sf p}^i}{a_n\sigma_{i,n}^2}(1-|\beta_{i,n}|\alpha_{\sf p}^{j_{i,n}})} \right)
	-\sum_{k=1}^{N_{\sf p}}{\left( \boldsymbol{\Delta}_k\alpha_{\sf p}^{-k}\sum_{i=k}^{N_{\sf p}}{\frac{\omega_{i,n}}{a_n\sigma_{i,n}^2}\alpha_{\sf p}^{i}}\right) }
	+\sum_{i=1}^{n}{\alpha_{\sf p}^{n-i}\boldsymbol{\Delta}_i} \\
	& +\sum_{k=1}^{n}{\left( \boldsymbol{\Delta}_k\alpha_{\sf p}^{-k}\sum_{i=k-1}^{n-1}{\frac{\omega_{i,n}}{a_n\sigma_{i,n}^2}|\beta_{i,n}|\alpha_{\sf p}^{i+1}}\right) }
	+\sum_{k=1}^{N_{\sf p}-1}{\left( \boldsymbol{\Delta}_k\alpha_{\sf p}^{-k}\sum_{i=k+1,i>n}^{N_{\sf p}}{\frac{\omega_{i,n}}{a_n\sigma_{i,n}^2}|\beta_{i,n}|\alpha_{\sf p}^{i-1}}\right) }\\
	&=\mathbf{h}_0\left( \alpha_{\sf p}^n-\sum_{i=1}^{N_{\sf p}}{\frac{\omega_{i,n}\alpha_{\sf p}^i}{a_n\sigma_{i,n}^2}\left( 1-|\beta_{i,n}|\alpha_{\sf p}^{j_{i,n}}\right) } \right)
	+
	\sum_{k=n+1}^{N-1}{
		\left[ 
		\boldsymbol{\Delta}_k\frac{\alpha_{\sf p}^{-k}}{a_n} \left(
		\sum_{i=k+1}^{N_{\sf p}}{\frac{\omega_{i,n}}{\sigma_{i,n}^2}|\beta_{i,n}|\alpha_{\sf p}^{i-1}}
		-\sum_{i=k}^{N_{\sf p}}{\frac{\omega_{i,n}}{\sigma_{i,n}^2}\alpha_{\sf p}^{i}}
		\right) 
		\right] }
	\\
	&  {+}\sum_{k=1}^{n}{ \!
		\left[ \!
		\boldsymbol{\Delta}_k\frac{\alpha_{\sf p}^{-k}}{a_n} \! \! \left(
		\sum_{i=k{-}1}^{n-1}{\! \! \frac{\omega_{i,n}}{\sigma_{i,n}^2}|\beta_{i,n}|\alpha_{\sf p}^{i+1}}
		{-}\!
		\sum_{i=k}^{N_{\sf p}}{\frac{\omega_{i,n}}{\sigma_{i,n}^2}\alpha_{\sf p}^{i}}
		{+}\sum_{i=n+1}^{N_{\sf p}}{\frac{\omega_{i,n}}{\sigma_{i,n}^2}|\beta_{i,n}|\alpha_{\sf p}^{i-1}}
		\right) 
		\right] }
	{+}\sum_{k=1}^{n}{\boldsymbol{\Delta}_k\alpha_{\sf p}^{n-k}}
	{-}\boldsymbol{\Delta}_N\frac{\omega_{N,n}}{a_n\sigma_{N,n}^2}\mathds{1}_{n<N}.
	\end{split}
	\end{equation}
	\setcounter{equation}{\value{tempequationcounter_1}}
	\hrulefill	
\end{figure*}
\addtocounter{equation}{1}

As SNR goes to infinity, $\text{lim}_{\rho \to \infty}\sum_{i=1}^{N_{\sf p}}{|\xi_{i,n}^{\sf c}|^2}=0$, the channel evolutionary noises contribute negligible power to the CEE. This completes the proof.

\section{Proof of Proposition 2}
Substituting $\mathbf{y}_{i,n}$ from (\ref{eq: xyB}) into (\ref{eq: c_tilde}), note that $\mathbf{y}_i=\mathbf{h}_ix_i+\mathbf{B}_i\mathbf{c}+\mathbf{w}_i$, and after some manipulations, we have
\begin{equation} \label{eq: res_noise}
\begin{split}
\tilde{\mathbf{c}}_n &=\mathbf{c}+\mathbf{D}_n^{-1}\sum_{i=1}^{N_{\sf p}}{\mathbf{G}_{i,n}\left( \mathbf{h}_ix_i+\mathbf{w}_i\right) },
\end{split}
\end{equation}
where $\mathbf{K}_n=\left( \sum_{i=1}^{N_{\sf p}}{x_{i,n}^*\boldsymbol{\Sigma}_{i,n}^{-1}\mathbf{B}_{i,n}}\right)^H\mathbf{A}_n^{-1}, \mathbf{J}_{i,n}=\left( \mathbf{B}^H_{i,n}-\mathbf{K}_nx_{i,n}^*\right) \boldsymbol{\Sigma}_{i,n}^{-1}$, and
\begin{equation}
\mathbf{G}_{i,n}=\left\{
\begin{array}{ll}
\mathbf{J}_{i,n} , &i{=}1,n,N_{\sf p}\\
\mathbf{J}_{i,n}
-\mathbf{J}_{i-1,n}\beta_{i-1,n} , &n>i>1\\
\mathbf{J}_{i,n}
-\mathbf{J}_{i+1,n}\beta_{i+1,n} , &n<i<N_{\sf p}.
\end{array}
\right.
\end{equation}

The second term in (\ref{eq: res_noise}) represents the estimation error of $\mathbf{c}$ at position $n$,
and it is independent of $\mathbf{c}$. This completes the proof for the first part of the proposition.

Additionally, it can be seen that the estimation error is a linear combination of zero-mean Gaussian random variables, hence it also has zero mean. Therefore, the estimation is unbiased.
The covariance matrix of the residual interference at position $n$th is
\begin{equation} \label{eq: covar_c}
\begin{split}
&(\sigma^2_{\sf h}{+}\sigma^2)
\E_{(\mathbf{x},\mathbf{B})}
\left[
\mathbf{B}_n\mathbf{D}_n^{-1}\left( \sum_{i=1}^{N_{\sf p}}{\mathbf{G}_{i,n}\mathbf{G}_{i,n}^H}\right) \mathbf{D}_n^{-1}\mathbf{B}_n^H \right]\!\!{+}\\
&\sigma^2_{\sf h} \sum_{i \neq j}^{}{\E_{(\mathbf{x},\mathbf{B},\mathbf{h}_i,\mathbf{h}_j\!)}\!\!
	\left[ \!
	\mathbf{B}_n\mathbf{D}_n^{-1}\!\!
	\left(\!
	\mathbf{G}_{i,n}\mathbf{h}_i\mathbf{h}_j^H\mathbf{G}_{j,n}^Hx_ix_j^* \!
	\right) \!
	\mathbf{D}_n^{-1}\mathbf{B}_n^H
	\right]}\\
&{=}(\sigma^2_{\sf h}{+}\sigma^2)
\E_{(\mathbf{x},\mathbf{B})}
\left[
\mathbf{B}_n\mathbf{D}_n^{-1}\left( \sum_{i=1}^{N_{\sf p}}{\mathbf{G}_{i,n}\mathbf{G}_{i,n}^H}\right) \mathbf{D}_n^{-1}\mathbf{B}_n^H \right]\!\!{+}\\
& \; \sigma^2_{\sf h}
\sum_{i \neq j}^{}{\!
	\alpha_{\sf p}^{|i{-}j|}\E_{(\mathbf{x},\mathbf{B})} \left[ \mathbf{B}_n\mathbf{D}_n^{-1}\left(\mathbf{G}_{i,n}\mathbf{G}_{j,n}^Hx_ix_j^*\right) \mathbf{D}_n^{-1}\mathbf{B}_n^H \right]
}.
\end{split}
\end{equation}

Deriving the closed-form expression for the covariance matrix is tedious. Hence, we will prove Proposition 2 by using
the following arguments. First, note that, $\mathbf{D}_n$ is Hermitian, positive definite, and in quadratic order of 
interfering matrices $\mathbf{B}_i, i=1,...,N$. Second, $\mathbf{B}_n\mathbf{G}_{i,n}$ is also in quadratic order of interfering matrices. Therefore, the expected covariance matrix is in a fractional function form with total zero-th order of $\mathbf{B}_i$. 
As a result, the residual interference  power is bounded as we increase the interference power to infinity. Furthermore, 
if
all interfering channel coefficients for different antennas have identical value, the residual interference power is \textit{completely} independent of the interference power.
Since estimation errors for EICs at any symbol position $n$ are finite, the overall estimation error for EICs is also finite.
This completes the proof.

\section{Proof of Proposition 3}
Since the expression of the power of residual interference in (\ref{eq: covar_c}) contains $\sigma^2_{\sf h}$, it does not vanish as $\rho \to \infty$.
As $\rho$ goes to infinity, we have

\begin{equation}
\begin{split}
\omega_{i,n}
&{\to} \left\{
\begin{array}{ll}
\!\! \alpha_{\sf p}, \!\!\!\! &i=n \pm 1\\
\!\! 1,              \!\!\!\! & i=n\\  
\!\! 0,              \!\!\!\! &\text{otherwise}\\  
\end{array}
\right.\!\!\!,
|\beta_{i,n}|
\to \left\{
\begin{array}{ll}
\!\! 0,              \!\!\!\! &i=n,n\pm 1\\
\!\! \alpha_{\sf p}, \!\!\!\! &\text{otherwise}\\
\end{array}
\right.\!\!\!,
\\
\sigma_{i,n}^2
&{\to} \left\{
\begin{array}{ll}
\sigma^2,  &i=n\\
1-a^{2}_{\sf p}, &\text{otherwise}
\end{array}
\right.,
\\
\mathbf{A}_n&{\to} \rho\mathbf{I}_{N_{\sf r}},
\\
\mathbf{D}_n
& {\to} \mathbf{B}_{n}^H\mathbf{B}_{n}+\sum_{i\neq n}^{N_{\sf p}}{\frac{\mathbf{B}_{i,n}^H\mathbf{B}_{i,n}}{1-\alpha_{\sf p}^2}}
+2\frac{\alpha_{\sf p}^2}{1-\alpha_{\sf p}^2}\mathbf{B}_n^H\mathbf{B}_n\\
&\;
-\frac{\alpha_{\sf p}}{1-\alpha_{\sf p}^2}\sum_{i=n \pm 1}^{}{\left( x_i^*x_n\mathbf{B}^H_{n}\mathbf{B}_{i}+
	x_ix_n^*\mathbf{B}_i^H\mathbf{B}_{n}
	\right) }\\
& {\to}  N_{\sf r}N_{\sf p}\frac{1+\alpha_{\sf p}^2}{1-\alpha_{\sf p}^2} \mathbf{I}_L,
\\
\mathbf{K}_n&\to x_n\mathbf{B}_n^H,
\\
\mathbf{J}_{n,n}
&{\to} \mathbf{B}^H_n\frac{1+\alpha_{\sf p}^2}{1-\alpha_{\sf p}^2}-\frac{x_n^*\alpha_{\sf p}}{1-\alpha_{\sf p}^2}\sum_{i=n\pm 1}^{}{x_i\mathbf{B}_i^H},
\\
\mathbf{J}_{i,n}
&{\to} \frac{\mathbf{B}_{i,n}^H-x_nx_{i}^*\alpha_{\sf p}\mathbf{B}_n^H}{1-\alpha_{\sf p}^2}, i\neq n,
\\
\mathbf{J}_{i,n}\mathbf{J}_{i,n}^H
& {\to} N_{\sf r}\left( \frac{1+|\beta_{i,n}|^2+\omega_{i,n}^2}{\sigma_{i,n}^4}\right)\mathbf{I}_L\\
&{\to}\frac{ N_{\sf r}\left( 1+\alpha_{\sf p}^2\right) }{(1-\alpha_{\sf p}^2)^2}\mathbf{I}_L , i \neq n,
\\
\mathbf{G}_{i,n}\mathbf{G}_{i,n}^H
& {\to} \mathbf{J}_{i,n}\mathbf{J}_{i,n}^H{+}
|\beta_{i\pm 1,n}|^2\mathbf{J}_{i\pm 1,n}\mathbf{J}_{i\pm 1,n}^H {+}
\frac{2N_{\sf r}|\beta_{i\pm 1,n}|^2}{\sigma_{i,n}^2\sigma_{i\pm 1,n}^2}\mathbf{I}_L\\
& {\to}  N_{\sf r}\frac{ 1+\alpha_{\sf p}^4+4\alpha_{\sf p}^2 }{(1-\alpha_{\sf p}^2)^2}\mathbf{I}_L, i \neq n.
\end{split}
\end{equation}

Upon having these asymptotic values, we substitute these values into (\ref{eq: covar_c}) to arrive at the residual interference 
power limit stated in the proposition, note that the computation of $\E_{(\mathbf{x},\mathbf{B})} \left[ \mathbf{B}_n\mathbf{D}_n^{-1}\left(\mathbf{G}_{i,n}\mathbf{G}_{j,n}^Hx_ix_j^*\right) \mathbf{D}_n^{-1}\mathbf{B}_n^H \right], i \neq j,$ is done similarly.
This completes the proof.

\bibliographystyle{IEEEtran}
\bibliography{main}

\end{document}